\pdfoutput=1
\documentclass[sigconf,nonacm,balance=false]{acmart}
\setcopyright{none}
\acmDOI{}
\acmISBN{}
\acmBooktitle{}
\renewcommand\footnotetextcopyrightpermission[1]{}

\usepackage{amsthm}
\usepackage[inline]{enumitem}
\usepackage{caption,subcaption}
\usepackage{balance}
\usepackage[binary-units]{siunitx}
\DeclareSIUnit{\nothing}{\relax}
\usepackage{graphicx}
\usepackage{pifont}
\usepackage{multirow}
\usepackage{color}
\usepackage[normalem]{ulem} % Defines strikeout Breaks bibtex
\usepackage{xspace}
\usepackage{cleveref}
\usepackage{comment}
\usepackage{soul}
\usepackage{ulem}
\usepackage{changepage}
\usepackage{thmtools}
\usepackage{float}
\makeatletter
\let\cassandra@addcontentsline\addcontentsline
\renewcommand{\addcontentsline}[3]{%
  \def\cassandra@tocfile{#1}%
  \def\cassandra@loefile{loe}%
  \ifx\cassandra@tocfile\cassandra@loefile
  \else
    \cassandra@addcontentsline{#1}{#2}{#3}%
  \fi
}
\makeatother

\usepackage{tikz,pgfplots,pgfplotstable}
\usetikzlibrary{positioning,arrows.meta,fit,backgrounds,shapes.geometric,calc,shapes.misc}

\newtheorem{lemma}{Lemma}

\newtheorem{definition}{Definition}

\providecommand{\texorpdfstring}[2]{#1}
\DeclareRobustCommand{\newshepherd}[1]{\texorpdfstring{#1}{#1}}
\DeclareRobustCommand{\new}[1]{\texorpdfstring{#1}{#1}}
\newcommand{\newcaptionshepherd}[1]{\caption{#1}}
\newcommand{\newcaption}[1]{\caption{#1}}

\newcommand{\Name}[1]{{#1}}

\newcommand{\BFT}{BFT}

\newcommand{\PoA}{PoA}
\newcommand{\PoR}{PoR}
\newcommand{\RC}{RC}

\newcommand{\SignMessage}[2]{\langle#1\rangle_{#2}}
\newcommand{\Hash}[1]{\texttt{H}(#1)}

\newcommand{\Proposal}{\phi}

\DeclareSIUnit\k{k}
\DeclareSIUnit\ms{ms}
\DeclareSIUnit\GB{GiB}
\DeclareSIUnit\B{B}
\DeclareSIUnit{\txn}{txn}
\DeclareSIUnit{\batch}{batch}
\definecolor{lightgreen}{RGB}{34,139,34}
\definecolor{lightyellow}{RGB}{218,112,214}
\definecolor{slateblue}{RGB}{123,104,238}
\definecolor{lightblue}{RGB}{30,144,255}
\definecolor{lightBrown}{RGB}{188,143,143}
\definecolor{deepGreen}{RGB}{0,128,0}
\definecolor{deepBlue}{RGB}{0,0,255}
\definecolor{deepPurple}{RGB}{128,0,128}
\definecolor{fakeGreen}{RGB}{102,205,170}
\definecolor{Maroon}{RGB}{210,105,30}
\definecolor{lightred}{RGB}{255,106,106}
\definecolor{SandyBrown}{RGB}{244,164,96}
\definecolor{SandyGreen}{RGB}{200,30,0}
\definecolor{SandyRed}{RGB}{255,164,96}

\newcommand{\RN}[1]{%
  \textup{\expandafter{\romannumeral#1}}%
}

\usepackage{algorithmic}

\newenvironment{myprotocol}{
    \hrule
    \footnotesize
    \smallskip
    \algsetup{linenosize=\footnotesize}
    \begin{algorithmic}[1]
        
        \newcommand{\SPACE}{\item[]}
        \newcommand{\TITLE}[2]{\item[] \textbf{\underline{##1}} (##2) \textbf{:}\\[2pt]}
        \makeatletter
            \newcommand{\EVENT}[1]{\STATE \textbf{event} ##1 \textbf{do}\begin{ALC@g}}
            \newcommand{\ENDEVENT}{\end{ALC@g}}
        \makeatother
        
        \makeatletter
            \newcommand{\FUNCTION}[2]{\STATE \textbf{function} \Name{##1}(##2) \textbf{do}\begin{ALC@g}}
            \newcommand{\ENDFUNCTION}{\end{ALC@g}}
        \makeatother
}{
    \end{algorithmic}%
    \hrule
}

\newcommand{\sysname}{\textsc{Cassandra}\xspace}

\newcommand{\Shaokang}[1]{{\color{blue}{\textbf{Shaokang:} #1}}}
\newcommand{\Dakai}[1]{\textcolor{cyan}{Dakai: {#1}}}
\newcommand{\bheading}[1]{{\vspace{4pt}\noindent{\textbf{#1}}}}
\newcommand{\iheading}[1]{{\vspace{2pt}\noindent{\textit{#1}}}}

\newcommand{\etal}{\textit{et al.}\xspace}

\AtBeginDocument{\clubpenalty=10000\widowpenalty=10000\displaywidowpenalty=10000}
\AtBeginEnvironment{thebibliography}{\interlinepenalty=10000}
\pgfplotsset{
    scaled y ticks = false,
    every axis/.append style={
        ylabel near ticks,
        xlabel near ticks,
        mark size=2.5pt,
    },
}

 \pgfplotstableread{
Id	Node	Throughput Latency		
1	16	860033.209	0.2240631798
2	32	780523	0.3968395
3	48	676317.791	0.5696158202
4	64	584894.9153	0.7423921404
5	104	484415.7062	0.9151684606
}\dataAutoBahnHS

\pgfplotstableread{
    Id	Node	Throughput Latency		
1	16	453997.71	0.0883795288
2	32	384063.311	0.241629605
3	48	344503.051	0.3948796812
4	64	247360.029	0.5481297574
5	104	201212.6975	0.6013798336
}\dataCassandra

\pgfplotstableread{
Id	Node	Throughput Latency
1	16	867655.2936	0.3322403382
2	32	776488	0.4336465
3	48	680308.3277	0.5693890449
4	64	571357.5693	0.8663965353
5	104	440541.0864	0.9849708886
}\dataTusk

 \pgfplotstableread{
Id	Node	Throughput Latency
1	16	880549.0548	0.1937866611
2	32	798966	0.3680308
3	48	677642.8776	0.4871401167
4	64	557867.3407	0.6046358673
5	104	441991.252	0.8013125951
}\dataRCC

\pgfplotstableread{
    Id	Node	Throughput Latency
1	16	889701.0955	0.3148159132
2	32	805212	0.3474001
3	48	672891.3962	0.4968343721
4	64	577167.1549	0.6608971581
5	104	477095.3053	0.7519063876
}\dataCassandraAuto

\pgfplotstableread{
Id	Node	Throughput Latency
1	16	478797.5138	0.2024533695
2	32	335311	0.2080577
3	48	282880.2363	0.2821068892
4	64	191058.6392	0.4960937629
5	104	137591.5005	0.7359205228
}\dataHS

\pgfplotstableread{
Id	Node	Throughput Latency
1	16	473635.6849	0.03237633365
2	32	389464	0.0531276
3	48	327759.3151	0.1387886635
4	64	243193.9635	0.2546301327
5	104	180469.2785	0.455381399
}\dataPBFT

\pgfplotstableread{
Id	Node	Throughput Latency
1	16	852554.9543	0.1618783063
2	32	728622	0.4225555
3	48	630030.0457	0.6832326937
4	64	508962.0913	0.9439098874
5	104	398518.137	1.204587081
}\dataSpotless

\pgfplotstableread{
Id	Node	Throughput Latency
1	16	301473.6625	0.0244425875
2	32	233269.7188	0.04511461739
3	48	96170.18987	0.06005084615
4	64	64691.54646	0.2658743077
5	104	12871.68986	0.318678437
}\dataZzy

\newcommand{\EvalBatch}[4]{\begin{tikzpicture}
    \begin{axis}[title={\normalsize \vspace{4cm} #1},ylabel={\footnotesize #2},xlabel={\footnotesize Batch Size},
            xtick={1,2,3,4,5},
            xticklabels={200, 400, 600, 800},
            width=130,
            height=100,
            grid=both,
            enlarge x limits=0.1,
            bar width=3pt,
            xmajorgrids=true,
            legend style={font=\footnotesize,  draw=none, only marks, legend columns=3
            },
            legend pos={south east},
            legend to name={batchmain},
            ymin=0,
            y tick label style={
                /pgf/number format/.cd,
                scaled y ticks=base 10:#4,
                fixed,
                /tikz/.cd,
                font=\scriptsize
            },
            x tick label style={
                font=\scriptsize
            }
        ]

        \addplot[color=slateblue] table[x={Id}, y={#3}] {\dataCassandraBatch};
        \addplot[color=lightyellow] table[x={Id}, y={#3}] {\dataTuskBatch};
        \addplot[color=lightblue] table[x={Id}, y={#3}] {\dataRCCBatch};

        \legend{\sysname, 
                Tusk,
                RCC}
    \end{axis}
\end{tikzpicture}}

\newcommand{\EvalBatchTPS}{\EvalBatch{}{Throughput (tps)}{Throughput}{-3}}
\newcommand{\EvalBatchLat}{\EvalBatch{}{Latency (s)}{Latency}{0}}

\pgfplotstableread{
Id	Batch Throughput Latency
1	200	3455957.476	3.219937058
2	400	4644405.031	4.58955323
3	600	5092333.662	5.974640878
4	800	5224242.857	7.183034779
}\dataCassandraBatch

\pgfplotstableread{
Id	Batch Throughput Latency
1	200	2806176.933	1.96492799
2	400	4506217.211	3.143919248
3	600	4861292.143	4.317621663
4	800	5203379.582	5.293434375
}\dataTuskBatch

\pgfplotstableread{
Id	Batch Throughput Latency
1	200	2780790.576	1.663515313
2	400	4521196.51	2.669992897
3	600	4870103.175	3.566625385
4	800	5369153.821	4.547684141
}\dataRCCBatch

\pgfplotstableread{
Id LThroughput Percent
1 715040	99.53610176
2 167440		44.0060698
3 209280		56.72131148 
4 168880		64.35483871
5 270080		46.16613419
6 167440		33.64010475
7 776960	106.2626263
8 727680	100.9801721
9 736720	98.29405163
10 729280	100
}\dataRecoveryOneThreeNoPre

 \pgfplotstableread{
Id LThroughput Percent
1 715040	64.35483871
2 167440		27.02702703
3 209280		32.25806452 
4 168880		60.60606061
5 270080		58.82352941
6 167440		29.41176471
7 776960	12.941176471
8 727680	53.57142857
9 736720	72.4137931
}\dataRecoveryLoss

\pgfplotstableread{
Id LThroughput Percent
1 715040	63.40694006
2 167440		0
3 209280		0 
4 168880		56.72131148
5 270080		58.18452381
6 167440		0
7 776960	0
8 727680	67.09265176
9 736720	65.55023923
}\dataRecoveryHSLoss

\pgfplotstableread{
Id LThroughput TB TBL  
1	736720	1841.8				10.84690069
2	0	0				0
3	0	0				0
4	0	0				0
5	8948000	22370				14.44927764
6	760320	1900.8				10.89239103
7	702160	1755.4				10.7775841
8	727680	1819.2				10.82908844
9	726960	1817.4				10.82766027
}\dataRecoveryHalf

\pgfplotstableread{
Id	LThroughput	TB				TBL
1	724960	1812.4				10.82368568
2	0	0				0
3	0	0				0
4	0	0				0
5	0	0				0
6	7177600	17944				14.13121391
7	1435200	3588				11.80896417
8	725040	1812.6				10.82384487
9	724640	1811.6				10.82304873
}\dataRecoveryHalfHalf

\pgfplotstableread{
Id	LThroughput	TB				TBL
1	746960	1867.4				10.86681527
2	0	0				0
3	0	0				0
4	0	0				0
5	0	0				0
6	1020768	2551.92				11.31736739
7	738160	1845.4				10.84971785
8	706880	1767.2				10.78724961
9	742720	1856.8				10.85860271
}\dataRecoveryHalfHalfHalf						

\pgfplotstableread{
Id LThroughput LThroughputLog 
1 808320	21.5 
2 0	        0
3 0	        0
4 0	        0
5 0	        0
6 0	        0
7 815360	21.5
8 820960	21.5
9 815360	21.5
10 803120	 21.5
}\dataRecoveryRCC

 \pgfplotstableread{
Id	LThroughput	TB				 TBL
1	808320	2020.8				10.98071083
2	0	0				0
3	0	0				0
4	0	0				0
5	840240	2100.6				11.03658575
6	863680	2159.2				11.07628117
7	776160	1940.4				10.92213837
8	815360	2038.4				10.99322147
9	823120	2057.8				11.00688706
}\dataRecoveryRCCHalf

\pgfplotstableread{
Id LThroughput TB  TBL
1	820790.5763	2051.976441		11.00279845
2	0	0				0
3	0	0				0
4	0	0				0
5	821196.5101	2052.991275				0
6	790103.1751	1975.257938				10.94782534
7	799153.8211	1997.884553				10.9642575
8	796619.5745	1991.548936				10.95967521
9	797128.505	1992.821262				10.9605966
}\dataRecoveryRCCHalfHalf

\pgfplotstableread{
Id LThroughput Percent
1 845600	100
2 0	        0 
3 0	        0 
4 0	        0 
5 0	        0 
6 0    100
7 867440	100
8 869440	100
9 823200	100
10 795360	 100
}\dataRecoveryAutoBahn

 \pgfplotstableread{
Id	LThroughput	TB	TBL			
1	715600	1789				10.80493767
2	0	0				0
3	0	0				0
4	0	0				0
5	7103520	17758.8				14.11624648
6	751440	1878.6				10.8754422
7	739600	1849				10.85252951
8	745360	1863.4				10.86372168
9	749520	1873.8				10.87175126
}\dataRecoveryAutoBahnHalf

\pgfplotstableread{
Id LThroughput TB TBL
1	720868.0193	1802.170048				10.81551943
2	0	0				0
3	0	0				0
4	0	0				0
5	0	0				0
6	7478340	18695.85				14.19043044
7	2134800	5337				12.3818133
8	722082	1805.205				10.81794696
9	715152	1787.88				10.80403419
}\dataRecoveryAutoBahnHalfHalf

\pgfplotstableread{
Id LThroughput TB TBL
1	820868.0193	2052.170048				11.00293457
2	0	0				0
3	0	0				0
4	0	0				0
5	0	0				0
6	1042656	2606.64				11.34797563
7	801840	2004.6				10.96909867
8	788080	1970.2				10.94412637
9	838960	2097.4				11.03438631
}\dataRecoveryAutoBahnHalfHalfHalf

\pgfplotstableread{
Id LThroughput Percent
1 823680	95.73221757
2 0	        0 
3 275072.233	31.97027348
4 0	        0 
5 275072.233	        31.97027348
6 0     0
7 275072.233	31.97027348
8 874080	101.5899582
9 858880	99.82333798
10 860400	 100
}\dataRecoveryAutoBahnHS

 \pgfplotstableread{
Id LThroughput TB  TBL
1	763680	1909.2				10.89875253
2	0	0				0
3	0	0				0
4	0	0				0
5	7021600	17554				14.09951219
6	694800	1737				10.76238204
7	739920	1849.8				10.85315358
8	739760	1849.4				10.85284158
9	746160	1865.4				10.86526931
}\dataRecoveryAutoBahnHSHalf

\pgfplotstableread{
Id LThroughput TB  TBL
1	684898.4615	1712.246154				10.7416744
2	0	0				0
3	0	0				0
4	0	0				0
5	0	0				0
6	9183680	22959.2				14.48678475
7	1004649.913	2511.624783				11.29440524
8	720868.0193	1802.170048				10.81551943
9	747827.5042	1869.568761				10.86848982
}\dataRecoveryAutoBahnHSHalfHalf

\pgfplotstableread{
Id LThroughput TB  TBL
1	884898.4615	2212.246154				11.11129621
2	0	0				0
3	0	0				0
4	0	0				0
5	0	0				0
6	985041.6232	2462.604058				11.26596897
7	865041.6232	2162.604058				11.07855384
8	820868.0193	2052.170048				11.00293457
9	847827.5042	2119.568761				11.04955505
}\dataRecoveryAutoBahnHSHalfHalfHalf

\pgfplotstableread{
Id LThroughput Percent
1 200320	98.62150453
2 0	        0
3 0	        0
4 25360        32.48523041
5 0	       0
6 0    0
7 20240	32.099
8 240960	103.8597873
9 245360	106.0259945
10 203120	 100
}\dataRecoveryHS

 \pgfplotstableread{
Id LThroughput LThroughputLog  Latency
1 867360 19.72627138647312 3.143919248
2 0  13 6.143919248
3 233792 13.8 6.143919248
4 0 13.751544059089099 6.143919248
5 233792 14.216745858195306 6.143919248
6 0 13.948367231584678 6.143919248
7 565680 18.828979397135278 4.143919248
8 838880 19.77776900401376 3.143919248
9 829448 19.644292877059453 3.143919248
10 838560 19.677554488525278 3.143919248
}\dataRecoveryTusk
\pgfplotstableread{
Id LThroughput TB  TBL
1	809328.7887	2023.321972				10.9825102
2	0	0				0
3	0	0				0
4	0	0				0
5	813151.5407	2032.878852				10.98930853
6	817437.0675	2043.592669				10.99689195
7	834997.7143	2087.494286				11.02755653
8	764063.3113	1910.158278				10.89947647
9	844503.0508	2111.257627				11.04388692
}\dataRecoveryTuskHalf

\pgfplotstableread{
Id LThroughput TB  TBL
1	834997.7143	2087.494286				11.02755653
2	0	0				0
3	0	0				0
4	0	0				0
5	795892.2222	1989.730556				0
6	709332.1053	1773.330263				10.79224553
7	756929.8641	1892.32466				10.88594391
8	771816.0656	1929.540164				10.91404136
9	771816.0656	1929.540164				10.91404136
}\dataRecoveryTuskHalfHalf

\pgfplotstableread{
Id Node Throughput Latency 
1	16	1290073.0526	1.124062932
2	32	1003113.3953	2.255148894
3	48	991521.6901	2.556771524
4	64	955135.6782	2.270197949
}\dataMultiPaxos

\pgfplotstableread{
Id Node Throughput Latency 
1	16	2077495.522	0.186499
2	32	2065606.619	0.4185923265
3	48	2114853.318	0.5957755257
4	64	2139342.599	0.7118323401
}\dataCassCFT

\pgfplotstableread{
Id Node Throughput Latency
1	16	351995.702	4.920654876
2	32	180259.227	5.907966586
3	48	101667.664	5.980074809
4	64	90779.941	6.760651736
}\dataFailCass

\pgfplotstableread{
Id Node Throughput Latency
1	16	1461103.495	1.355580618
2	32	1367478.242	2.36744854
3	48	1329477.127	3.644126989
4	64	1177671.91	4.228489377
}\dataFailTusk

\pgfplotstableread{
Id Node Throughput Latency
1	16	438045.921	1.001692402
2	32  305537.7521	2.191622927
3	48	268476.673	2.86012736
4	64	183292.42	3.29763889
}\dataFailCassCFT

\pgfplotstableread{
Id Node Throughput Latency
1	16	1030713.615	0.9193019814
2	32	900029.0419	1.167954781
3	48	862862.5907	2.22503514
4	64	770988.7351	2.73160612
}\dataFailPaxos

\pgfplotstableread{
Id  Throughput 
1	77857.14284
2	80000
3	80000
4	80000
5	80000
6	80000
7	80000
8	80000
9	80000
10	80000
11	80000
12	80000
13	80000
14	80000
15	80000
16	01
17	01
18	01
19	01
20	01
21	01
22	01
23	01
24	01
25	01
26	01
27	01
28	01
29	01
30	01
31	01
32	01
33	01
34	01
35	01
36	01
37	01
38	01
39	01
40	01
41	01
42	01
43	01
44	01
45	01
46	01
47	01
48	01
49	01
50	01
51	01
52	01
53	01
54	01
55	01
56	01
57	01
58	01
59	01
60	01
61	01
62	01
63	01
64	01
65	01
66	01
67	01
68	01
69	789333.333
70	80000
71	80000
72	80000
73	80000
74	79951.0204
75	80000
76	80000
77	80000
78	80000
79	79800
80	80000
81	80000
82	80000
83	79995.91836
84	80000
85	80000
86	80000
87	80000
88	80000
89	80000
90	80000
91	80000
92	80000
93	80000
94	80000
95	80000
96	80000
97	80000
98	80000
99	80000
100	80000
}\dataRecoveryNoPOA

\pgfplotstableread{
Id  Throughput 
1	77500
2	80000
3	80000
4	79882.75864
5	80000
6	80000
7	80000
8	80000
9	80000
10	80000
11	80000
12	80000
13	80000
14	80000
15	80000
16	01
17	01
18	01
19	01
20	01
21	01
22	01
23	01
24	01
25	01
26	01
27	01
28	01
29	01
30	01
31	01
32	01
33	01
34	01
35	01
36	01
37	01
38	01
39	01
40	01
41	01
42	01
43	01
44	01
45	01
46	01
47	01
48	01
49	01
50	01
51	01
52	01
53	01
54	01
55	01
56	01
57	01
58	01
59	01
60	01
61	01
62	01
63	01
64	01
65	01
66	01
67	378064.516
68	80000
69	80000
70	80000
71	80000
72	79661.22448
73	80000
74	80000
75	80000
76	80000
77	79600
78	80000
79	80000
80	80000
81	80000
82	80000
83	80000
84	80000
85	80000
86	80000
87	80000
88	80000
89	80000
90	80000
91	80000
92	80000
93	80000
94	80000
95	80000
96	80000
97	80000
98	80000
99	80000
100	80000
}\dataRecoveryPOA

\pgfplotstableread{
Id  Throughput 

1	77777.77776
2	80000
3	80000
4	80000
5	80000
6	80000
7	80000
8	80000
9	80000
10	80000
11	80000
12	80000
13	80000
14	80000
15	80000
16	80000
17	60000
18	40000
19	40000
20	40000
21	40000
22	40000
23	40000
24	40000
25	40000
26	40000
27	40000
28	40000
29	40000
30	40000
31	40000
32	40000
33	40000
34	40000
35	40000
36	40000
37	40000
38	40000
39	40000
40	40000
41	40000
42	40000
43	40000
44	40000
45	40000
46	40000
47	40000
48	40000
49	40000
50	40000
51	40000
52	40000
53	40000
54	40000
55	40000
56	40000
57	40000
58	40000
59	40000
60	40000
61	40000
62	40000
63	40000
64	40000
65	40000
66	40000
67	742533.333
68	80000
69	80000
70	80000
71	80000
72	79785.567
73	79444.89796
74	80000
75	80000
76	80000
77	79734.69388
78	80000
79	80000
80	80000
81	80000
82	80000
83	80000
84	80000
85	80000
86	80000
87	80000
88	80000
89	79857.14284
90	80000
91	80000
92	80000
93	80000
94	80000
95	80000
96	80000
97	80000
98	80000
99	80000
100	80000
}\dataRecoveryPOR

\pgfplotstableread{
Id  Throughput 
 1 0.56
 2 3.33
 3 5.24
}\dataRecoveryTimeNoPOA

\pgfplotstableread{
Id  Throughput 
 1 0.55
 2 1.02
 3 2.02
}\dataRecoveryTimePOA

\AtBeginDocument{%
  
  }

\begin{document}

%\title{\sysname: A Layered Consensus with Partial Progress}
%\title{\sysname: A Robust Partitionable View Synchronization}

%\title{\sysname: A Robust Partitionable View Synchronization for Consensus with Partial Progress}

\title{\sysname: Consensus with Partial Progress via Robust Partitionable View Synchronization}

%Partial Progress Conjecture under Asynchrony
%\title{Did we miss P In CAP? \\ Partial Progress Conjecture under Asynchrony}

\author{Shaokang Xie}
\affiliation{%
  \institution{University of California, Davis}
  \city{}
  \country{}}
\email{skxie@ucdavis.edu}

\author{Dakai Kang}
\affiliation{%
  \institution{University of California, Davis}
  \city{}
  \country{}}
\email{dakang@ucdavis.edu}

\author{Junchao Chen}
\affiliation{%
  \institution{University of California, Davis}
  \city{}
  \country{}}
\email{jucchen@ucdavis.edu}

\author{Suyash Gupta}
\affiliation{%
  \institution{University of Oregon}
  \city{}
  \country{}}
\email{suyash@uoregon.edu}

\author{Daniel P. Hughes}
\authornote{In loving memory of Daniel P. Hughes, a valued friend and colleague, whose foundational contributions, ideas, and passion continue to inspire and shape this work.}
\affiliation{%
  \institution{Radix DLT Ltd.}
  \city{}
  \country{}}
\email{dan@radixfoundation.org}

\author{Mohammad Sadoghi}
\affiliation{%
  \institution{University of California, Davis}
  \city{}
  \country{}}
\email{msadoghi@ucdavis.edu}
\nocite{sourcecode, cloudflare_outage1, cloudflare_outage2, aws_useast_outage}

%%\renewcommand{\shortauthors}{Trovato et al.}

\begin{abstract}

Replicated databases and permissioned blockchains rely on Byzantine Fault-Tolerant (BFT) consensus 
to maintain a consistent transaction order across replicas. 
These protocols preserve safety even under asynchrony, 
as they commit a transaction only after agreement among a strong quorum of replicas. 
During network partitions, however, when no strong quorum is reachable, they lose liveness and cannot make useful progress. 

In this paper, we present \sysname{}, a consensus protocol that enables partial progress without sacrificing safety. 
\sysname{} achieves this through a two-tier certification framework that decouples availability from commitment, 
allowing each partition to extend its own chain and reconcile these chains once the network is restored. 
\new{To support this, \sysname{} introduces a pacemaker that advances views without requiring a strong quorum and calibrates each replica's timeout off the critical path.}

Our evaluation results show that \sysname remains competitive with state-of-the-art BFT protocols under stable conditions, \new{sustaining \mbox{$900$K} TPS at 16 replicas and \mbox{$480$K} TPS at 104 replicas, with latency ranging from \mbox{$0.31$s} at 16 replicas to \mbox{$0.75$s} at 104 replicas.}
\new{Under severe partitions, \sysname{} maintains non-zero speculative throughput through \PoA{}-backed progress, preserving work that can be reconciled once connectivity is restored.}
\end{abstract}

\maketitle

\section{INTRODUCTION}
\new{Decentralized systems offer their clients an immutable ledger even in the presence of faulty or Byzantine behavior.
To achieve this, they employ a Byzantine Fault-Tolerant (\BFT{}) consensus protocol that allows multiple distributed parties to maintain a single, globally consistent transaction order over a shared state~\cite{blockchain-book,sadoghi2025problems,bc-processing}. 
Because these parties often hold the same data, they are termed replicas.}
%
%In decentralized deployments, this role is typically fulfilled by Byzantine Fault-Tolerant (\BFT{}) consensus protocols, which preserve agreement among replicas even in the presence of faulty or malicious behavior.
\new{State-of-the-art \BFT{} protocols, which have served as the foundational building blocks of permissioned blockchains, provide strong fault tolerance and achieve high throughput and low latency under stable network conditions~\cite{pbftj,hotstuff,hs1,poe,flexi-trust}.}

\begin{figure}[t]
\centering
\vspace*{0mm}
\resizebox{\columnwidth}{!}{
\begin{tikzpicture}[
    % Style Definitions
    replica/.style={circle, draw=black, thick, fill=blue!10, minimum size=1.1cm, font=\bfseries},
    bad_link/.style={draw=black!60, thin},
    cross_mark/.style={draw=red, ultra thick, cross out, minimum size=0.4cm},
    partition_line/.style={red!80!black, ultra thick, decorate, decoration={zigzag, segment length=6mm, amplitude=1.2mm}},
    log_box/.style={draw=black!80, fill=green!20, minimum width=0.4cm, minimum height=0.4cm, inner sep=0pt},
    arrow_good/.style={-{Stealth[scale=0.8]}, green!60!black, ultra thick},
    arrow_bad/.style={-{Stealth[scale=0.8]}, red!80, ultra thick},
    progress_sign/.style={circle, draw=green!60!black, fill=green!60!black, minimum size=0.4cm, path picture={\fill[white] (path picture bounding box.west) -- (path picture bounding box.north east) -- (path picture bounding box.south east) -- cycle;}},
    forbidden_sign/.style={circle, draw=red, line width=1pt, minimum size=0.4cm, path picture={\draw[red, line width=1pt] (path picture bounding box.south west) -- (path picture bounding box.north east);}}
]

    % === Left Side: Traditional BFT (Stalled) ===
    \begin{scope}[local bounding box=leftside]
        \node[replica] (R1) at (0, 2) {R1};
        \node[replica] (R2) at (2.5, 2) {R2};
        \node[replica] (R3) at (0, 0) {R3};
        \node[replica] (R4) at (2.5, 0) {R4};

        % Internal Communication Arrows
        \draw[arrow_good, <->] (R1) -- (R3);
        \draw[arrow_good, <->] (R2) -- (R4);
        \draw[arrow_bad, <->] (R1) -- (R2); 
        \draw[arrow_bad, <->] (R3) -- (R4); 
        \draw[arrow_bad, <->] (R1) -- (R4); 
        \draw[arrow_bad, <->] (R2) -- (R3); 
        % \draw[arrow_bad, <->] (R1) -- (R2) node[midway, cross_mark, scale=0.7] {}; 

        % Partition Line
        \draw[partition_line] (1.25, 2.8) -- (1.25, -1);
        \node[red!80!black, font=\bfseries, align=center] at (0.9, 3) [xshift=0.4cm] {Network Partition};

        % Bottom Log (Stalled)
        \foreach \x in {0,1,2,3} {
            \node[log_box, fill=gray!5] (el\x) at (0.5 + \x*0.4, -1.8) {};
        }
        \node[above=0.03cm of el2, font=\small, xshift=-0.2cm] (logL) {\textbf{Global Log}};
        \draw[red!70, ultra thick] (0.3, -1.8) -- (1.9, -1.8) node[midway, below=0.2cm, font=\fontsize{8.5}{9}\selectfont\bfseries] {COMMIT BLOCKED};

        \node[forbidden_sign, right=0.2cm of el3] (fs) {};
        \node[right=0.08cm of fs, font=\fontsize{8.5}{9}\selectfont, red!70] {FROZEN};

        \node[above=0.6cm of R1, xshift=1.25cm, font=\large\bfseries] {Traditional BFT (Stalled)};
        \node[below=1cm of logL, font=\small, align=center] {
            Requires \textbf{$>2/3$ Quorum} (3 nodes) \\ Now available: 2 nodes \\ 
            \textcolor{red}{\textbf{Status: STALLED (Throughput = 0)}}
        };
    \end{scope}

    % Central Divider
    \draw[dashed, thick, gray!50] (4.2, 3.5) -- (4.2, -4);

    % === Right Side: \sysname (Partial Progress) ===
    \begin{scope}[shift={(6,0)}]
        \node[replica] (S1) at (0, 2) {R1};
        \node[replica] (S2) at (2.5, 2) {R2};
        \node[replica] (S3) at (0, 0) {R3};
        \node[replica] (S4) at (2.5, 0) {R4};

        \draw[partition_line] (1.25, 2.8) -- (1.25, -1);
        \node[red!80!black, font=\bfseries, align=center] at (0.9, 3) [xshift=0.4cm] {Network Partition};

        % Internal Communication Arrows
        \draw[arrow_good, <->] (S1) -- (S3);
        \draw[arrow_good, <->] (S2) -- (S4);
        \draw[arrow_bad, <->] (S1) -- (S2); 
        \draw[arrow_bad, <->] (S3) -- (S4); 
        \draw[arrow_bad, <->] (S1) -- (S4); 
        \draw[arrow_bad, <->] (S2) -- (S3); 
        % \draw[arrow_bad, <->] (R1) -- (R2) node[midway, cross_mark, scale=0.7] {}; 

        % PoA Logs (Group A)
        \foreach \x in {0,1,2} {
            \node[log_box] (b\x) at (0.2-\x*0.4, -1.8) {$\checkmark$};
        }
        \node[above=0.03cm of b1, font=\small] {\textbf{\shortstack{\em Group A \\ Local Log}}};
        \node[progress_sign, right=0.2cm of b0] (ps1) {};
        \node[below=0cm of ps1, font=\fontsize{8.5}{9}\selectfont, green!80!black, yshift=-0cm] {ACTIVE};

        % PoA Logs (Group B)
        \foreach \x in {0,1,2} {
            \node[log_box] (c\x) at (1.9+\x*0.4, -1.8) {$\checkmark$};
        }
        \node[above=0.03cm of c1, font=\small, align=left] {\textbf{\shortstack{\em Group B \\ Local Log}}};
        \node[progress_sign, right=0.2cm of c2] (ps2) {};
        \node[below=0cm of ps2, font=\fontsize{8.5}{9}\selectfont, green!80!black, yshift=-0cm] {ACTIVE};

        \node[above=0.6cm of S1, xshift=1.25cm, font=\large\bfseries] {\sysname (Partial Progress)};
        \node[below=2cm of S3, xshift=1.25cm, font=\small, align=center] {
            Requires \textbf{$>1/3$ Quorum} (2 nodes) \\ Now available: 2 nodes \\ 
            \textcolor{green!60!black}{\textbf{Status: PARTIAL PROGRESS}}
        };
    \end{scope}

\end{tikzpicture}
}
\vspace{-7mm}
\caption{Comparison between traditional BFT protocols and \sysname under a network partition. In a 2-2 partition scenario, traditional protocols stall as they cannot reach the required $2/3$ quorum, whereas \sysname maintains partial progress using its Proof of Availability (PoA) mechanism.}
\label{fig:partition_comparison}
\vspace{-1mm}
\end{figure}

In practice, however, networks are not always stable.
Transient network partitions, delay asymmetries, and intermittent connectivity can disrupt communication among replicas.
High-impact production incidents reported by major technology providers such as Amazon and Cloudflare demonstrate that 
network partitions are not merely a theoretical concern but a recurring operational reality with severe consequences~\cite{alquraan2018,solana-failure,aws_useast_outage,cloudflare_outage1,cloudflare_outage2}. 
For example, a major Amazon EC2/EBS outage in the \texttt{us-east-1} region revealed how misrouted traffic and aggressive recovery storms can effectively partition storage clusters, rendering data unavailable for extended periods and causing widespread downtime across SaaS providers and financial platforms, with aggregate losses reaching millions of dollars~\cite{aws_useast_outage}. 
Similarly, Cloudflare has experienced outages in which routing anomalies, overloaded or impaired interconnects, and regional reachability failures created partitions between cloud regions~\cite{cloudflare_outage1,cloudflare_outage2}. 
The impact is even more pronounced in decentralized and blockchain systems:
network congestion–induced outages in Solana have led to millions of dollars in losses for individual traders~\cite{solana-failure}.
\new{Moreover, partitions are not always accidental: routing-layer attacks, such as BGP hijacks, have been shown to isolate nodes and split the Bitcoin network into disconnected components~\cite{apostolaki2017hijacking}.}
\new{These incidents highlight partition-like failures as a practical risk for geo-distributed services.}

\new{For \BFT{} replicated systems, such partition-like failures create a quorum-stall problem: protocols preserve \emph{safety}---that is, no two correct replicas commit conflicting transactions---but fail to process new requests when no connected component contains a sufficiently large quorum.}
\new{This tension between safety and availability is a consequence of the CAP theorem
(\emph{Consistency}, \emph{Availability}, and \emph{Partition tolerance}), which states that in the presence of network 
partitions, a distributed system cannot simultaneously guarantee both consistency and availability~\cite{cap-theorem,gilbert2012perspectives}.}
\new{Consequently, BFT systems designed for mostly reliable networks can lose availability under partitions, even though they continue to preserve safety.}

\BFT{}-based decentralized systems can stall under network partitions because the underlying consensus protocol
suffers from the following structural dependencies:

\paragraph{\textcolor{black}{\textbf{C1: Strong-Quorum Dependence.}}}
\new{Traditional \BFT{} protocols operate under the assumption that the system has at least $n=3f{+}1$ replicas,
where $f$ denotes the number of faulty replicas.
To guarantee that all correct replicas agree on a single transaction order, these protocols require a \emph{strong quorum}---at least $n{-}f$ ($2f{+}1$)---of replicas to certify a proposal.
Under a network partition, replicas are divided into multiple \emph{connected components} and replicas in different components are unable to communicate with each other (refer to Fig.~\ref{fig:partition_comparison}).
If no connected component contains a strong quorum (has fewer than $n{-}f$ replicas), 
then no proposal can form a certificate that justifies non-conflicting commitment.}
%Consequently, the protocol cannot produce new committed blocks, resulting in zero throughput.

\paragraph{\textbf{C2: Designated-Leader Dependence.}}
Classic \BFT{} protocols (such as PBFT~\cite{pbft} and HotStuff~\cite{hotstuff}) operate in rounds, each with
a designated leader responsible for determining the order of transactions. 
However, under a network partition, the leader is unreachable for all but the replicas in its connected component.
When progress stalls, replicas assume that the leader has failed and initiate the recovery protocol to replace the leader.
However, the recovery protocol itself requires participation of a quorum of replicas.
If no partition contains a strong quorum, replicas repeatedly attempt leader replacement without ever making progress.

\paragraph{\bf \new{C3: Round-Overlap under Varying Network Delay.}}
\new{Certificate formation also requires replicas to remain in the same round long enough to exchange proposals and votes.
Thus, progress depends on sufficient \textit{round overlap} among replicas in the same connected component.
However, under partial synchrony, the effective delay bound is unknown before stabilization and may vary across stabilized periods~\cite{dwork1988consensus,pbft}.
A timeout that is too small may cause replicas to leave a round before messages arrive, while an overly conservative timeout increases latency.}

\paragraph{\textbf{Consequence: Zero Useful Progress.}}
Together, strong-quorum dependence, leader dependence, and round-overlap requirements imply that existing \BFT{} protocols make no useful progress during a partition unless some connected component can both gather the required quorum and keep enough replicas synchronized in the same round.
Replicas may continue retransmitting messages and triggering recovery attempts, but no new transactions can be committed, resulting in zero throughput.
\newshepherd{This limitation is shared by performance-oriented BFT protocols, although they optimize different aspects of the good case.
Concurrent and multi-leader designs~\cite{rcc,spotless,iss,ladon,orthrus,hydra,mirbft,dqbft} improve throughput; the HotStuff family~\cite{hotstuff,hs2,hs1} pipelines or shortens the commit path; DAG or dissemination-decoupled protocols~\cite{keidar2021dagrider,narwhal-tusk,spiegelman2022bullshark,autobahn} improve throughput by separating data dissemination from ordering; and optimistic fast-path protocols~\cite{zyzzyva,fab,sbft} reduce latency through a fast certification path.
Despite these optimizations, without reaching a strong quorum, their progress towards ordering or finality stalls.
}

In this paper, we present {\bf \sysname}, a novel consensus protocol that revisits the practical implications of CAP in decentralized replicated systems~\cite{blockchain-book}.
Rather than attempting to achieve full liveness during partitions, \sysname{} enables \emph{Partial Liveness}, allowing \textit{partial progress} while strictly preserving safety during network partitions.
\sysname{} addresses the above challenges through three complementary mechanisms.
First, \sysname{} introduces a {\em two-tier certification mechanism} for each transaction:
\begin{itemize}[nosep]
\item \emph{Proof of Availability} (\PoA{}), a certificate formed by collecting votes from a weak quorum of $f{+}1$ replicas.
\item \emph{Proof of Reliability} (\PoR{}), a certificate formed by collecting votes from a strong quorum of $n{-}f$ replicas.
\end{itemize}
A \PoA{} certificate enables replicas to make {\em partial progress} during a network partition. 
However, to guarantee consistency, \sysname{} requires \PoR{} certificates to commit transactions.

Second, \sysname{} does not designate any replica as ``the leader''; instead, it allows all replicas to propose transactions.
Each replica then applies our deterministic {\em proposal priority rule} to independently select the {\em strongest proposal}
among those it has received.
As a result, even during a network partition, every replica can always identify a locally strongest proposal to vote for.

\new{Third, \sysname{} deploys a novel \textbf{decoupled pacemaker} for round or view synchronization.
%\footnote{In this paper, we use ``round'' and ``view'' interchangeably: a round in \sysname{} corresponds to the logical view in prior BFT protocols~\cite{hotstuff}. We keep the term ``view synchronization'' to follow prior terminology.}
Traditional pacemakers do perform round synchronization, but they face the following challenge:}

\paragraph{\bf \new{C4: Pacemaker Constraints.}}
\new{Existing pacemakers guarantee sufficient round overlap among a strong quorum of $2f{+}1$ replicas, 
which may be unavailable under network partitions. 
Achieving this requires the pacemaker to perform two tasks: 
advancing replicas' rounds (or views), and 
calibrating their timers so that these rounds overlap. 
However, the pacemaker performs both tasks on the critical consensus path, which degrades system performance.

\sysname{} decouples these tasks as follows: 
logical round advancement remains certificate-driven on the critical consensus path, 
while timeout calibration runs in the background—each replica adjusts its local round timeout independently, without blocking consensus progress.
However, one challenge still remains unresolved.
}
%\new{These mechanisms allow replicas to make partition-local progress while preserving global safety.
%Such progress can create divergent histories that must be reconciled after recovery:}

\paragraph{\bf \new{C5: Convergence.}}
\new{Although \PoA{} certificates enable progress under network partitions,
they also allow different connected components to advance along distinct \PoA{}-backed branches.}
Because replicas can observe different subsets of proposals, conflicting \PoA{} certificates can coexist, and 
Byzantine replicas may further equivocate across competing branches.
Consequently, once the network is restored, \sysname{} must ensure that these divergent partial histories can be reconciled into a single committed order.

\new{\sysname{} addresses this challenge in two ways. 
First, its deterministic proposal-priority rule makes reconciliation \textit{implicit}: 
once connectivity is restored, correct replicas identify the highest-priority branch as the strongest, vote for it, and extend it with new \PoR{}s. 
Second, the \textit{decoupled pacemaker} enables lagging replicas to catch up whenever they observe higher-round certificates. 
As a result, once a strong quorum becomes reachable again, 
all correct replicas—whether previously stalled or progressing within separate partitioned components—reconcile their histories by extending the same strongest branch.
}

%\new{With these designs, \sysname{} preserves partial progress inside connected $f{+}1$ partitions, enables fast catch-up and implicit reconciliation once $2f{+}1$ replicas become connected again, and adapts to unknown and time-varying delay through background timeout calibration.}

\new{Our evaluation demonstrates that \sysname{} sustains partial progress during network partitions while remaining competitive under stable network conditions.
Under steady-state, \sysname{} achieves $900$K TPS at 16 replicas and $480$K TPS at 104 replicas, retaining $53.6\%$ of its throughput, 
while its latency increases from $0.31$s at 16 replicas to $0.75$s at 104 replicas.}
During network partitions, \sysname achieves a `0-to-1' breakthrough by sustaining \textit{Continuous \PoA{} Generation} with a speculative throughput of $500$K TPS even when global consensus stalls (\Cref{fig:eval_protocols}(c)). 
This accumulated partial progress enables a rapid \textit{recovery burst} upon network restoration, allowing the system to efficiently reconcile divergent histories and resume peak performance quickly.

% Specifically, we make the following contributions:
% \begin{itemize}[leftmargin=*]

%     \item \textbf{Two-tier certification for partition-aware partial progress.}
%     We introduce a two-tier certification framework that decouples recoverable partial progress (\PoA) from final commitment (\PoR).

%     \item \textbf{Partitionable leadership under partitions.} 
%     We eliminate the dependence of designated leader by allowing replicas to continuously propose and to locally converge on the strongest available proposal within each connected component by a deterministic voting rule.

%     \item \textbf{Certificate-driven round synchronization.} 
%     We design a pacemaker that advances rounds through certificate-backed progress and proposal-driven fast catch-up. Replicas may either enter a higher round using local certificate, or upon receiving a valid higher-round proposal. This allows the system to self-align continuously without relying on explicit epoch-boundary synchronization.

%     \item \textbf{Eventual reconciliation after partition.} 
%     We show how divergent \PoA{}-backed histories accumulated during partitions can be safely reconciled after connectivity is restored, ensuring a globally consistent transaction ordering once the network is restored. 

%     \item \textbf{Empirical evaluation under partitions.} We conduct an extensive performance analysis of \sysname. Our results demonstrate that \sysname remains competitive with state-of-the-art protocols under stable network conditions, while sustaining non-zero progress during severe partitions.

% \end{itemize}

Specifically, we make the following contributions:
\begin{itemize}[nosep,wide]

    \item %\textbf{Two-tier certification for partition-aware partial progress.}
    We introduce a two-tier certification mechanism that separates partial progress (\PoA{}) from final commitment (\PoR{}), enabling progress within connected components while preserving safety.

    \item %\textbf{Partitionable chain construction without designated leaders.} 
    We eliminate reliance on a designated leader by allowing all replicas to propose continuously and 
    to independently vote for the strongest proposal according to a deterministic proposal-priority rule. 
    This design enables local chain extension during partitions.

    \item %\textbf{Decoupled pacemaker for round synchronization and timeout calibration.}
    We design a decoupled pacemaker that enables round advancement even with weak quorums, while calibrating local timeouts in the background, off the critical path.
    %This enables partial progress under $f{+}1$ connectivity, fast catch-up under $2f{+}1$ connectivity, and adaptation to unknown and time-varying network delay}.

    \item %\textbf{Implicit reconciliation after network recovery.} 
    We show that divergent \PoA{}-backed histories accumulated during partitions can reconcile once connectivity is restored.
    %Through deterministic proposal selection and renewed \PoR{} formation, the system converges back to a single globally consistent committed order.

    \item %\textbf{Empirical evaluation under adversarial partitions.} 
    We conduct an extensive evaluation of \sysname{} under both stable and partitioned network conditions.
    Our results demonstrate that \sysname{} remains competitive with state-of-the-art protocols in synchronous settings while sustaining non-zero useful progress during severe partitions.

\end{itemize}

\section{SYSTEM MODEL}
\label{subsec:model}

% Distributed databases must operate under the persistent threat of network partitions: high-impact production incidents underscore that this is not a merely theoretical concern~\cite{alquraan2018,solana-failure}. 
% In financial decentralized systems and databases, the stakes are even higher: 
% for instance, outage due to network congestion at Solana led to millions of dollars of losses for individual traders~\cite{solana-failure}.
%We consider a decentralized system with $n$ replicas. 
%The system must tolerate Byzantine faults, where faulty replicas may have arbitrary malicious behaviors.
%Replicas communicate over an unreliable network subject to message delays and transient partitions. Accordingly, we assume a partially synchronous network model with authenticated communication.
%\vspace{1.5mm}
\paragraph{\textbf{Adversary Model.}}
We consider a distributed system consisting of $n$ replicas, each assigned a unique identifier \(i \in \{1,\ldots,n\}\). We denote the replica with identifier \(i\) by \(p_i\).
The system tolerates up to \(f\) Byzantine replicas, where \(n \geq 3f + 1\). 
Byzantine replicas, also referred to as faulty replicas, may exhibit arbitrary or malicious behaviors that deviate from the consensus protocol. The remaining replicas are \emph{correct} (or \emph{honest}) and strictly follow the protocol.
The system processes requests from clients, and we allow any number of clients to be arbitrarily faulty.

\vspace{1.5mm}
\noindent \textbf{Authenticated Communication.}  
We assume an authenticated communication model based on standard public-key cryptography~\cite{cryptobook}. Each client and replica possesses a public--private key pair and uses digital signatures to authenticate messages. A sender signs a message using its private key, and any recipient can verify the signature using the corresponding public key. We denote by \(\SignMessage{m}{p_i}\) a digital signature over message \(m\) produced by replica \(p_i\).

Correct replicas accept only well-formed messages that carry valid signatures. We assume that Byzantine replicas are computationally bounded and cannot forge signatures or break the cryptographic primitives used by correct replicas.

We further assume the existence of a collision-resistant hash function \(\Hash{\cdot}\), meaning it is computationally infeasible to find two distinct inputs \(x \neq y\) such that \(\Hash{x} = \Hash{y}\).

\vspace{1.5mm}
\noindent \textbf{Threshold-Signature Setup.}
\new{In addition to ordinary digital signatures, \sysname{} uses threshold signatures for compact certificates: an $(f{+}1,n)$ scheme for weak-quorum certificates and the deterministic threshold coin and a $(2f{+}1,n)$ scheme for strong-quorum \PoR{} certificates.
We assume that the corresponding threshold verification keys and replica signing shares are generated before any runtime partition, either by a trusted dealer or by executing a standard DKG protocol~\cite{pedersen1991threshold,gennaro1999secure} while the replicas are fully connected.
}

\vspace{1.5mm}
\noindent {\bf Network Model.}  \label{subsubsec:network_model}
%We consider a system of \(n\) replicas interconnected by authenticated, point-to-point channels.  
%The links are \emph{reliable} in the sense that messages are neither corrupted nor forged, but they may be delayed or dropped.
We assume the \emph{partial synchrony} model of Dwork \etal~\cite{dwork1988consensus}.  
Existing \BFT{} protocols assume the existence of a single Global Stabilization Time (\textsf{GST}) after which the network remains permanently synchronous~\cite{hotstuff,hs1}.
In contrast, we incorporate a more realistic setting that allows recurring network partitions.  
Specifically, we assume an unbounded sequence of stabilized periods, each beginning at a stabilization point \(\textsf{GST}_k\), where \(\textsf{GST}_1 < \textsf{GST}_2 < \cdots\).
During each stabilized period, the network is synchronous for a sufficiently long but possibly finite interval, meaning that messages between correct replicas are delivered within an \textit{unknown} bound $\Delta$. 
Outside such synchronous periods, we make no global timing or connectivity assumptions: the network may be asynchronous, partially connected, or partitioned into multiple connected components.
%In particular, a partition is not assumed to be internally synchronous by default.
Our partial-progress guarantees are therefore conditioned on any connected component in which messages among mutually reachable correct replicas are delivered within $\Delta$.

\section{ACHIEVING PARTIAL PROGRESS UNDER NETWORK PARTITIONS}
\label{sec:design_overview}

Traditional BFT protocols cannot make progress during a network partition unless at least $2f{+}1$ replicas can communicate within one connected component;
we refer to such a component as a {\em strong-quorum component}.
\sysname{}, by contrast, aims to enable \emph{partial progress} while preserving safety even when no component contains a strong quorum.
It achieves this goal through the following design principles:

\begin{enumerate}[nosep,wide]
    \item \textbf{Eliminating futile recovery traffic.}
    When no component contains a strong quorum, it is wasteful for replicas to repeatedly trigger recovery mechanisms in an attempt to achieve global synchronization.
    Instead, they should continue issuing new proposals based on locally available connectivity.

    \item \textbf{Partitionable leaders.}
    During a network partition, not all replicas can necessarily reach the replica designated as leader. 
    Instead, each replica should take over this role itself, independently selecting the strongest proposal from those it has received.

    \item \textbf{Weak-quorum ordering.}
    During a network partition, replicas may be unable to reach a component containing a strong quorum.
    In that case, they should order transactions based on \emph{weak quorums}, 
    while recognizing that such decisions are non-final and must be reconciled once a strong quorum becomes available.
    
\item \textbf{Decoupled pacemaker.}
    \new{
    Traditional pacemakers are designed to synchronize a strong quorum of replicas, and they do so by placing all round-synchronization logic on the critical consensus path. 
    Because such a quorum may be unavailable during a network partition, 
    a pacemaker should instead be able to coordinate a weak quorum of replicas, while keeping much of this coordination off the critical path.
    }
\end{enumerate}

\subsection{Tools for Partial Progress}
\label{ss:tools-partial-progress}

To implement the aforementioned design principles, we introduce:
\begin{enumerate}[nosep,wide]
    \item \emph{Partitionable Leader Election}.
    \item \emph{Decoupled Pacemaker Design}.
\end{enumerate}
Next, we explain how these tools help \sysname{} enable partial progress during network partition.

\paragraph{\bf (1) Partitionable Leader Election}
Permitting leader election to proceed even when no strong quorum is reachable
allows each partition to make progress independently, without waiting for a dedicated leader. 
We achieve this goal through two new principles:  {\em two-tier certification} and {\em deterministic proposal priority}.

\paragraph{\normalfont\itshape Two-Tier Certification.}
Because no strong quorum is reachable during a network partition,
we allow replicas to make decisions based on weaker connected components. 
Consequently, to guarantee safety, we require two types of certificates.

\begin{itemize}[leftmargin=*]
  \item \emph{Proof of Availability (\PoA).}
    A \PoA{} requires votes from at least $f{+}1$ replicas and 
    guarantees that the proposal's payload is retrievable within the system, but it does not guarantee a final commit.

  \item \emph{Proof of Reliability (\PoR).}
  A \PoR{} requires votes from at least $n{-}f$ replicas.
  It certifies global consistency and prevents conflicts.
\end{itemize}

\paragraph{\normalfont\itshape Deterministic Proposal Priority Rule.}
Eliminating the need for replicas to wait for a designated leader means they must instead determine the next leader themselves,
and this determination must be deterministic so that \PoA{}s can be formed. 
Consequently, in each round, every replica broadcasts a proposal and 
independently identifies the leader from the proposals it receives.

Each proposal is ranked according to the recency of the certificates it carries—first by the highest referenced \PoR{}, and then by the highest referenced \PoA{}. 
Intuitively, proposals backed by more recent strong-quorum evidence are preferred, 
while weak-quorum evidence serves as a tie-breaking signal reflecting more recent partial progress. 
If multiple proposals remain tied within the same round, replicas resolve the tie using a deterministic threshold-coin mechanism derived from proposal signature shares. 
This guarantees that all correct replicas observing the same proposal set select the same \emph{strongest proposal}, enabling partition-local convergence without centralized leadership.

\paragraph{\normalfont\itshape Implicit Reconciliation on Network Recovery.}
During a network partition, replicas may form \PoA{} certificates for diverging chains within different components. 
However, because proposal selection is deterministic and prioritized by certificate recency, 
correct replicas converge on the same strongest proposal and vote for it as soon as the network recovers. 
New \PoR{}s then form on that branch, re-establishing a single, globally committed order. 
Reconciliation is thus not performed by an explicit merge protocol; 
it emerges naturally from deterministic proposal selection once certificate-driven round synchronization is restored.

\paragraph{\bf (2) Decoupled Pacemaker Design.}
\new{
Although partitionable leader election is necessary for progress, it is not sufficient. 
A successful election requires replicas to remain in the same {\em round} (or {\em view}) long enough to exchange proposals and
votes, so a mechanism for round synchronization is essential. 
Prior work employs a {\em pacemaker} that aligns a strong quorum of $n{-}f$ replicas on a common round, 
ensuring sufficient overlap for them to exchange messages. 
Moreover, the pacemaker runs on the critical consensus path and thus degrades consensus performance.

We achieve this through a novel decoupled pacemaker that separates the pacemaker's two tasks: 
{\em round advancement} and {\em timeout calibration}. 
Round advancement remains on the critical consensus path, while timeout calibration is now handled by a separate background mechanism. 
}

\paragraph{\normalfont\itshape Round Advancement.}
A replica advances from round $r$ to $r{+}1$ under two cases: 
(i) when a \PoR{} is formed, indicating successful certification by a strong quorum; or 
(ii) when, after round $r$ times out, a weak quorum of $f{+}1$ replicas requests a round advancement. 
In both cases, there is sufficient information for lagging replicas to jump immediately to the {\em highest certified round}. 
Convergence toward a common round is guaranteed because round numbers increase monotonically and 
replicas accept certificates only for strictly higher rounds. 
During synchronous periods, certificates propagate within $\Delta$, bounding divergence among correct replicas.

\new{
\paragraph{\normalfont\itshape Timeout Calibration.}
Another equally important role of the pacemaker is ensuring sufficient round overlap among replicas 
to facilitate the exchange of messages and proposals. 
This must happen even under unknown and time-varying delays, and it should remain off the critical path. 
We meet this goal not by forcing all replicas to synchronize on a single timeout value, 
but by allowing each replica to calibrate its own local timeout.
}

%This allows replicas to continue making partial progress inside connected $f{+}1$ partitions, while still enabling fast catch-up and convergence once $2f{+}1$ replicas become connected again.

\subsection{System Guarantees}
In a system of $n = 3f{+}1$ replicas, where up to $f$ replicas may be Byzantine, 
%\sysname{} aims to reach an agreement on the totally ordered ledger of client transactions.
\sysname{} offers the following guarantees:

\begin{itemize} [wide]
    \item \textbf{Safety:} If two correct replicas $p_i$ and $p_j$ commit two transactions $tx_1$ and $tx_2$ at the same position in the log, then $tx_1=tx_2$.
    \item \textbf{Liveness:} During any sufficiently long stabilized synchronous period, correct replicas continuously commit new proposals containing client transactions.
\end{itemize}

Moreover, \sysname offers {\em partial liveness} even during network partitions. 
Specifically, we introduce the following property:

\begin{itemize} [wide]
    \item \textbf{Partial Liveness:} 
    During network partitions, 
    any synchronous connected component containing at least $f{+}1$ correct replicas continues ordering new proposals containing client transactions.
\end{itemize}

\begin{figure}[t]
    \vspace{-2mm}
    \begin{tikzpicture}[>=stealth, thick]
    \useasboundingbox (0,0.5) rectangle (8.0,-1.9);

    \draw (0,0) -- (7.85,0);
    \node[left] at (0,0) {R1};
    \draw (0,-0.6) -- (7.85,-0.6);
    \node[left] at (0,-0.6) {R2};
    \draw (0,-1.2) -- (7.85,-1.2);
    \node[left] at (0,-1.2) {R3};
    \draw[color=red] (0,-1.8) -- (7.85,-1.8);
    \node[left, color=red] at (0,-1.8) {R4};

    \node[align=center, fill=white, text=red] at (0.2, -2.2) {\small{\textit{R4 is Faulty}}};

    \node[align=center, draw, fill=white, rounded corners=1pt, inner sep=2pt] at (0.5,0) {P1$_r$};
    \node[align=center, draw, fill=white, rounded corners=1pt, inner sep=2pt] at (0.5,-0.6) {P2$_r$};
    \node[align=center, draw, fill=white, rounded corners=1pt, inner sep=2pt] at (0.5,-1.2) {P3$_r$};
    \node[align=center, draw=red, text=red, fill=white, rounded corners=1pt, inner sep=2pt] at (0.5,-1.8) {P4$_r$};

    \draw[->] (0.8,0) -- (2.45,-0.6);
    \draw[->] (0.8,0) -- (2.45,-1.2);
    \draw[->] (0.8,0) -- (2.45,-1.8);
    \draw[->] (0.8,-0.6) -- (2.45,0);
    \draw[->] (0.8,-0.6) -- (2.45,-1.2);
    \draw[->] (0.8,-0.6) -- (2.45,-1.8);
    \draw[->] (0.8,-1.2) -- (2.45,0);
    \draw[->] (0.8,-1.2) -- (2.45,-0.6);
    \draw[->] (0.8,-1.2) -- (2.45,-1.8);
    \draw[->, red] (0.8,-1.8) -- (2.45,0);
    \draw[->, red] (0.8,-1.8) -- (2.45,-0.6);
    \draw[->, red] (0.8,-1.8) -- (2.45,-1.2);

    \node[align=center, draw, fill=white, rounded corners=1pt, inner sep=2pt] at (3.1,0) {Vote P2$_r$};
    \node[align=center, draw, fill=white, rounded corners=1pt, inner sep=2pt] at (3.1,-0.6) {Vote P2$_r$};
    \node[align=center, draw, fill=white, rounded corners=1pt, inner sep=2pt] at (3.1,-1.2) {Vote P2$_r$};
    \node[align=center, draw=red, text=red, fill=white, rounded corners=1pt, inner sep=2pt] at (3.1,-1.8) {Vote P4$_r$};

    \draw[->] (3.75,0) -- (5.4,-0.6);
    \draw[->] (3.75,0) -- (5.4,-1.2);
    \draw[->] (3.75,0) -- (5.4,-1.8);
    \draw[->] (3.75,-0.6) -- (5.4,0);
    \draw[->] (3.75,-0.6) -- (5.4,-1.2);
    \draw[->] (3.75,-0.6) -- (5.4,-1.8);
    \draw[->] (3.75,-1.2) -- (5.4,0);
    \draw[->] (3.75,-1.2) -- (5.4,-0.6);
    \draw[->] (3.75,-1.2) -- (5.4,-1.8);
    \draw[->, red] (3.75,-1.8) -- (5.4,0);
    \draw[->, red] (3.75,-1.8) -- (5.4,-0.6);
    \draw[->, red] (3.75,-1.8) -- (5.4,-1.2);

    \node[align=center, draw, fill=white, rounded corners=1pt, inner sep=2pt] at (6.04,0) {PoR-P2$_r$};
    \node[align=center, draw, fill=white, rounded corners=1pt, inner sep=2pt] at (6.04,-0.6) {PoR-P2$_r$};
    \node[align=center, draw, fill=white, rounded corners=1pt, inner sep=2pt] at (6.04,-1.2) {PoR-P2$_r$};
    \node[align=center, draw, fill=white, rounded corners=1pt, inner sep=2pt] at (6.04,-1.8) {PoR-P2$_r$};

    \draw[<->, blue, thick] (0.15,0.3) -- (3.1,0.3) node[midway, above, text height=1.5ex, text depth=0.25ex] {\small Proposal Exchange};
    \draw[<->, blue, thick] (3.1,0.3) -- (6.7,0.3) node[midway, above, text height=1.5ex, text depth=0.25ex] {\small Election and Certification};
    \draw[<-, blue, thick] (6.7,0.3) -- (8.3,0.3) node[midway, above, text height=1.5ex, text depth=0.25ex] {\small Next Round};

    \node[align=center, draw, fill=white, rounded corners=1pt, inner sep=2pt] at (7.2,0) {P1$_{r+1}$};
    \node[align=center, draw, fill=white, rounded corners=1pt, inner sep=2pt] at (7.2,-0.6) {P2$_{r+1}$};
    \node[align=center, draw, fill=white, rounded corners=1pt, inner sep=2pt] at (7.2,-1.2) {P3$_{r+1}$};
    \node[align=center, draw=red, text=red, fill=white, rounded corners=1pt, inner sep=2pt] at (7.2,-1.8) {P4$_{r+1}$};

    \node[align=center, fill=white] at (8.1,0) {$\dots$};
    \node[align=center, fill=white] at (8.1,-0.6) {$\dots$};
    \node[align=center, fill=white] at (8.1,-1.2) {$\dots$};
    \node[align=center, fill=white] at (8.1,-1.8) {$\dots$};

    \end{tikzpicture}
    \vspace{-1mm}
    \caption{A round of \sysname{}.} \label{fig:slow_path}
    \vspace{-5mm}

\end{figure}

\section{\MakeUppercase{\sysname}}
\label{sec:overview}
\sysname{} aims to drive BFT consensus among replicas while guaranteeing partial liveness even during periods of network partition. 
To this end, \sysname{} must ensure the following:
\begin{itemize}[nosep,wide]
    \item \emph{Partitionable chain construction.}
    This enables replicas within each connected component to keep proposing, allowing the chain to grow even when a strong quorum is unreachable. 
    \sysname{} achieves this through \emph{partitionable leader election}.

    \item \emph{Round synchronization.}
    This keeps replicas sufficiently synchronized to exchange messages and form certificates, and provides lagging replicas a mechanism to catch up. 
    \sysname{} achieves this through the \emph{decoupled pacemaker}.
\end{itemize}

\subsection{Partitionable Chain Construction} 
\label{subsec:chain_construction} \label{subsubsec:chain_growth}
\sysname{} proceeds in a sequence of rounds (or views). 
Each round consists of two phases: {\em proposal exchange} and {\em election and certification} (refer to Figure~\ref{fig:slow_path}). 
A replica enters a new round by following the pacemaker's round advancement rules (\S\ref{ss:tools-partial-progress}).

%\new{Each replica has a local round timeout parameter $\delta_i$, which can be dynamically adjusted by the pacemaker (\Cref{ssec:pacemaker}).
%}Under normal execution, the nominal round length is  $3\delta_i$: a $2\delta_i$ proposal-exchange window followed by a  $1\delta_i$ election window.
%A replica may also \textit{fast-forward} into a higher round upon receiving a valid higher-round proposal.
%In that case, it immediately joins the higher round and executes only the remaining local waiting needed before entering the election phase.

%\sysname constructs the chain in a partition-aware, round-based manner.
%In each round $r$, a replica disseminates a proposal
%\(
%\mathcal{P} \triangleq \langle r,\; \text{txs},\; parent,\; enterCert,\; 
%\widehat{PoR},\; \widehat{PoA} \rangle,
%\)
%where $r$ is the round number, $\text{txs}$ is a batch of client transactions, $parent$ is the hash of the parent proposal, $enterCert$ is the certificate that justifies entry into round $r$, and  $\widehat{PoR}$ / $\widehat{PoA}$ denote the highest known \PoR{} and the highest extending \PoA{}, respectively.
%Intuitively, $enterCert$ certifies that the sender is legally in round $r$, while  $\widehat{PoR}$ and $\widehat{PoA}$ summarize the strongest globally reliable and partition-local progress anchors known to the sender.
%Each replica maintains a local log $\mathcal{L}$ consisting of proposals and their associated certificates as illustrated in Fig.~\ref{fig:log_chain}.

\paragraph{Phase 1: Proposal Exchange.}
At the start of round $r$, each replica aggregates a set of pending client transactions into a proposal. 
Each proposal {\em extends} the last committed proposal;
each replica maintains a local log $\mathcal{L}$ of committed proposals (Fig.~\ref{fig:log_chain}) and 
sets the hash of the proposal at the tail as $parent$ and its corresponding certificate as the highest \PoR{}.
Additionally, the proposal includes the highest \PoA{} (if any) and a certificate, which proves that the sender is in round $r$. 

This exchange takes place during an {\em exchange window}, in which replicas collect round-$r$ proposals from one another. 
\new{
As in other partially synchronous protocols, \sysname{} assumes that each replica $i$ maintains a local timeout 
parameter $\delta_i$ (with $\delta_i \ge \Delta$) that delimits the phases within each round. 
However, in \sysname the value of $\Delta$ is unknown a priori, so we make the value of $\delta_i$ able to calibrate dynamically over time (\Cref{subsubsec:background-resync}). 
Within a connected component, we assume an exchange window of size $2\delta_i$, 
which should suffice for proposals from honest replicas to propagate under synchrony.
}

%% TODO
%If a replica receives a valid proposal for a higher round $r' > r$, it may immediately adopt round $r'$ instead of waiting for its current round to finish.
%This is safe because the received proposal carries a valid \textit{enterCert}, which already justifies entry into round $r'$.
%The lagging replica then broadcasts its own round-$r'$ proposal and joins the remaining execution of that round.

\begin{figure}[t]
\vspace{-4mm}
\centering
\begin{tikzpicture}[
    proposal/.style = {
      rectangle, draw=black, rounded corners=2pt,
      minimum width=1.3cm, minimum height=0.7cm,
      align=center, font=\scriptsize
    },
    por/.style  = {proposal, fill=blue!15},
    poa/.style  = {proposal, fill=green!15},
    rc/.style   = {proposal, fill=orange!15, dashed},
    record/.style = {
      draw, rounded corners=4pt, fill=none,
      thick,
      minimum width=8.1cm, minimum height=1.55cm
    },
    field/.style = {
      draw, rounded corners=2pt,
      minimum width=1.4cm, minimum height=0.7cm,
      align=center, font=\footnotesize
    }
  ]

  % ---------------- Outer Container ----------------
  \node[record] (rec) at (0.4, 2.0) {};

  % ---------------- Proposal Structure (Top) ----------------
  \node[font=\small\bfseries] (struct_label) at (0.4, 2.42) {Structure of Proposal $P_r$};
  \node[field] (round) at (-2.75, 1.75) {Round\\Number};
  \node[field] (parent) [right=1.3mm of round] {Transactions};
  \node[field] (justify) [right=1.3mm of parent] {enterCert};
  \node[field] (txdata) [right=1.3mm of justify] {Hash of\\Parent};
  \node[field] (metadata) [right=1.3mm of txdata] {Highest\\PoA \& PoR};
  % ---------------- Local Log (Bottom) ----------------
  \node[font=\small\bfseries] (log_label) at (0.4, 0.45) {Local Log (Chain of Certificates)};

  \node[por] (p1) at (-2.8, -0.2) {Proposal 1\\ \textbf{PoR}};
  \node[rc]  (p2) [right=1.5mm of p1] {Round 2\\ \textbf{RC}};
  \node[poa] (p3) [right=1.5mm of p2] {Proposal 3\\ \textbf{PoA}};
  \node[por] (p4) [right=1.5mm of p3] {Proposal 4\\ \textbf{PoR}};
  \node[por] (p5) [right=1.5mm of p4] {Proposal 5\\ \textbf{PoR}};
  \node[]    (dots) [right=1mm of p5] {...};

  % Visual linking for back-filling intuition
  \draw[->, >=stealth, thin, gray] (p5.west) -- (p4.east);
  \draw[->, >=stealth, thin, gray] (p4.west) -- (p3.east);
  \draw[->, >=stealth, thin, gray] (p3.west) -- (p2.east);
  \draw[->, >=stealth, thin, gray] (p2.west) -- (p1.east);

  \begin{pgfonlayer}{background}
    \node[
      draw, dotted,
      rounded corners=4pt,
      fill=gray!5,
      inner sep=2mm,
      fit=(p1)(p2)(p3)(p4)(p5)(dots)(log_label)
    ] (locallog_box) {};
  \end{pgfonlayer}

\end{tikzpicture}
\vspace{-3mm}
\caption{\sysname's Proposal and Local Log Chain.}
\label{fig:log_chain}
\vspace{-4mm}

\end{figure}

\vspace{1.5mm}
\iheading{Phase 2: Election and Certification.} 
Next, from all the valid proposals a replica receives, 
it deterministically selects a \textit{Strongest Proposal} and broadcasts its \textsc{Vote}.
% \Dakai{Shaokang, please move section 4 here.}
%
Concretely, each proposal $\mathcal{P}$ is ranked by its latest certificates (\PoR{} and \PoA{}),
\[
\mathcal{S}(\mathcal{P}) \triangleq
\Big(
Round( \mathcal{P}.\widehat{PoR})~,~
Round( \mathcal{P}.\widehat{PoA})
 \Big),
\]
where  $\mathcal{P}.\widehat{PoR}$ and $\mathcal{P}.\widehat{PoA}$ denote the highest-round \PoR{} and \PoA{} referenced by $\mathcal{P}$ and $\mathcal{S}(\mathcal{P})$ is the strongest proposal.
Replicas use the following lexicographic ordering to compare proposals:
\begin{enumerate}[nosep,wide]
   \item \textit{Highest PoR:} A proposal whose log contains a \PoR{} from a higher round is always preferred, as it represents a more recent globally reliable certificate.
   \item \textit{Highest PoA:} If two proposals have the same highest \PoR{} round, the one with a higher \PoA{} round is preferred, as it reflects more recent partial progress. 
\end{enumerate}

If multiple proposals remain tied in round $r$, replicas apply a deterministic common-coin tie-breaker.
To support this, every proposal carries the proposer's threshold signature share for round $r$
\new{(refer to \Cref{subsec:model})}.
Let $S_r$ be any set of $f{+}1$ valid threshold signature shares from distinct proposals in round $r$, and aggregating any valid set of $f{+}1$ shares yields the same round coin \(\mathcal{C}_r=\textsc{ThresholdAggregate}(S_r)\), and proposal
\(\mathcal{P}\) receives score
\(\mathrm{Score}(\mathcal{P})=H(\mathcal{C}_r\,\|\,\mathcal{P}.\textsf{sender})\).
The proposal with the highest $\mathrm{Score}(\cdot)$ is designated as the Strongest Proposal of round $r$, 
which we denote as $SP_r$.
Each replica then broadcasts a \textsc{Vote} for $SP_r$.
If a proposal gathers $Q = 2f{+}1$ votes, a \PoR{} is formed. 
If it times out before forming a \PoR{} but gathers $Q_{\textit{weak}} = f{+}1$ votes, a \PoA{} is formed.

%If no proposal gathers even a weak quorum before the timeout, replicas initiate round advancement via \textsc{WishNewRound}, from which a \textit{Round Certificate (RC)} is formed if $f{+}1$ \textsc{WishNewRound} messages are collected.

\begin{figure}[t]
\vspace{-2mm}
\begin{myprotocol}
\TITLE{Data Structure}{replica $p_i$}

\SPACE\textbf{Data structures \& State}
\SPACE\quad \textbf{Proposal Msg:} $P \triangleq \langle r,\; \text{txs},\; \text{parent},\; \text{enterCert},\; \widehat{PoR},\; \widehat{PoA} \rangle$
\SPACE\quad \textbf{Local State:} $r_i$ (round), $\mathcal{L}_i$ (log),  $\widehat{PoR}_i$, $\widehat{PoA}_i$, $lock_i$
\SPACE\quad \textbf{Buffer State:} $\mathcal{P}_i[r]$ (proposals), $\textit{voteSet}_i$ (votes), $\textit{Voted}_i$ (flag)
\SPACE\quad \textbf{Local Timers:} $\delta_i,\;\textit{ExchangeTimer},\; \textit{ElectionTimer}$
\SPACE

\TITLE{Helper Functions}{replica $p_i$}

\FUNCTION{\textsc{ValidEnterCert}}{$cert,\; r$}
\STATE \textbf{return} $cert$ is well-formed and one of:
\STATE\quad\quad\quad (i) $cert = \textit{genesis}$ and $r=1$, or
\STATE\quad\quad\quad (ii) $cert$ is \PoR{} with $\text{round}(cert)=r-1$, or
\STATE\quad\quad\quad (iii) $cert$ is \RC{} with $\text{round}(cert)=r-1$
\ENDFUNCTION
\SPACE

\FUNCTION{\textsc{SafeToVote}}{$P$}
\STATE \textbf{return} $P$ is valid and either
\STATE\quad\quad\quad (i) $P$ extends $lock_i$, or
\STATE\quad\quad\quad (ii) $P$ carries a \PoR{} from a strictly higher round
\ENDFUNCTION
\SPACE

\FUNCTION{\textsc{SelectStrongest}}{$\mathcal{P}$}
\STATE \textbf{return} the proposal in $\mathcal{P}$ with highest priority \hfill \COMMENT{\Cref{subsubsec:chain_growth}}
\ENDFUNCTION
\SPACE

\EVENT{\textsc{EnterRound}$(r,\; cert,\; t_x)$}
\STATE \textbf{if} $r \le r_i$ or \textsc{ValidEnterCert}$(cert,\; r)$ = False \textbf{then} return
\STATE Cancel $\textit{ExchangeTimer}[r_i]$;\; Cancel $\textit{ElectionTimer}[r_i]$
\STATE $r_i \gets r$;\; $\textit{Voted}_i[r_i] \gets$ False
\STATE \textbf{if} $\mathcal{P}_i[r_i]$ is uninitialized \textbf{then} $\mathcal{P}_i[r_i] \gets \emptyset$
\STATE \textbf{start} $\textit{ExchangeTimer}[r_i]$ \textbf{with timeout} $t_x$
\STATE $parent \gets \text{hash}(\textsc{TailOfLog}(\mathcal{L}_i))$
\STATE $P^\star \gets \textsc{NewProposal}(r_i,\; \textsc{txs},\; parent,\; cert,\; \widehat{PoR}_i,\; \widehat{PoA}_i)$
\SPACE \COMMENT{$\widehat{PoR}_i,\widehat{PoA}_i$ carry the highest known PoR and PoA}
\STATE Broadcast($\langle \textsc{Propose}(P^\star) \rangle$)
\STATE $\mathcal{P}_i[r_i] \gets \mathcal{P}_i[r_i] \cup \{P^\star\}$ \hfill \COMMENT{self-proposal}
\ENDEVENT
\SPACE

\FUNCTION{\textsc{UpdateHiCert}}{$cert$}
\STATE \textbf{if} $cert$ is \PoR{} \textbf{and} $\text{round}(cert) > \text{round}(\widehat{PoR}_i)$ \textbf{then} $\widehat{PoR}_i \gets cert$
\STATE \textbf{if} $cert$ is \PoA{} \textbf{and} $\text{round}(cert) > \text{round}(\widehat{PoA}_i)$ \textbf{then} $\widehat{PoA}_i \gets cert$
\ENDFUNCTION
\SPACE

\FUNCTION{\textsc{MaybeAdvanceLockAndCommit}}{$lock_i,\; cert$}
\STATE \COMMENT{Two-PoR Rule---lock on the 1st PoR, commit at the 2nd PoR}
\STATE \textbf{if} $cert$ and its ancestors form the required PoR chain \textbf{then}
\STATE\quad update $lock_i$ to the latest safely locked proposal
\STATE\quad commit older proposals according to the above rule
\ENDFUNCTION

\end{myprotocol}
\vspace{-2.5mm}
\caption{Utility and Helper Functions for \sysname} \label{fig:utility_helper}
\vspace{-4mm}
\end{figure}

\begin{figure}[t!]
\vspace{-2mm}
\begin{myprotocol}
\TITLE{Round Logic}{replica $p_i$}

\SPACE\textbf{Initialization}
\SPACE\quad initialize  $lock_i,\; \widehat{PoR}_i,\; \widehat{PoA}_i$ to \textit{genesis};\; $r_i \gets 0$
\SPACE\quad Call \textsc{EnterRound}$(1,\; \textit{genesis},\; 2\delta_i)$
\SPACE

\EVENT{Received valid $\langle \textsc{Propose}(P_r) \rangle$ \textbf{from} $p_j$}
\STATE \textbf{if} $P_r.\text{round} < r_i$ \textbf{or} \textsc{ValidEnterCert}$(P_r.enterCert,$ $P_r.\text{round})$ = False \textbf{then} return \hfill \COMMENT{Outdated or invalid}
\STATE \textsc{UpdateHiCert} $(P_r.\widehat{PoR})$;\; \textsc{UpdateHiCert}$(P_r.\widehat{PoA})$
\STATE \textsc{MaybeAdvanceLockAndCommit}$(lock_i,\; P_r.\widehat{PoR})$
\SPACE \COMMENT{Update local safety state before processing this proposal}
\STATE \textbf{if} $P_r.\text{round} > r_i$ \textbf{then}
\SPACE\quad \COMMENT{Fast-forward to an existing higher round}
\STATE\quad Call \textsc{pacemaker.ObserveHigherRound}$(P_r)$
\STATE \textbf{if} \textsc{SafeToVote}$(P_r)$ \textbf{then}
\STATE\quad $\mathcal{P}_i[P_r.\text{round}] \gets \mathcal{P}_i[P_r.\text{round}] \cup \{P_r\}$
\ENDEVENT
\SPACE

\EVENT{$\textit{ExchangeTimer}[r]$ \textbf{timeout} \COMMENT{Exchange phase ends}}
\STATE \textbf{if} $r \ne r_i$ \textbf{or} $\textit{Voted}_i[r]$ \textbf{then} return
\STATE $\mathcal{P}^{safe}_i[r] \gets \{ P \in \mathcal{P}_i[r] \mid \textsc{SafeToVote}(P) \}$
\STATE \textbf{if} $\mathcal{P}^{safe}_i[r] = \emptyset$ \textbf{then}
\STATE\quad \textsc{ExitRound}$(r,\; \perp)$;\; return
\STATE $P^{sp} \gets \textsc{SelectStrongest}(\mathcal{P}^{safe}_i[r])$
\STATE $\textit{voteSet}_i[r,\; \text{hash}(P^{sp})] \gets \textit{voteSet}_i[r,\; \text{hash}(P^{sp})] \cup \{p_i\}$
\STATE Broadcast($\langle \textsc{Vote}(r,\; \text{hash}(P^{sp})) \rangle$)
\STATE $\textit{Voted}_i[r] \gets$ True
\STATE \textbf{start} $\textit{ElectionTimer}[r]$ \textbf{with timeout} $\delta_i$
\ENDEVENT
\SPACE

\EVENT{Received valid $\langle \textsc{Vote}(r,\; h) \rangle$ \textbf{from} $p_j$ \COMMENT{Collect votes}}
\STATE \textbf{if} $r \ne r_i$ \textbf{then} return
\STATE $\textit{voteSet}_i[r,\; h] \gets \textit{voteSet}_i[r,\; h] \cup \{p_j\}$
\STATE $c \gets |\textit{voteSet}_i[r,\; h]|$
\STATE \textbf{if} $c \ge Q$ \textbf{then} \hfill \COMMENT{strong quorum $Q=2f{+}1$}
\STATE\quad $\langle \text{PoR},\; h \rangle \gets \textsc{BuildCert}(r,\; h,\; c,\; \text{PoR})$
\STATE\quad append $\langle \text{PoR},\; h \rangle$ to $\mathcal{L}_i$
\STATE\quad \textsc{UpdateHiCert}$(\langle \text{PoR},\; h \rangle)$
\STATE\quad \textsc{MaybeAdvanceLockAndCommit}$(lock_i,\; \widehat{PoR}_i)$
\STATE\quad \textsc{ExitRound}$(r,\; \langle \text{PoR},\; h \rangle)$
\ENDEVENT
\SPACE

\EVENT{$\textit{ElectionTimer}[r]$ \textbf{timeout} \COMMENT{Election phase ends}}
\STATE \textbf{if} $r \ne r_i$ \textbf{then} return
\STATE $h^\star \gets \arg\max_h |\textit{voteSet}_i[r,\; h]|$
\STATE $c^\star \gets |\textit{voteSet}_i[r,\; h^\star]|$
\STATE \textbf{if} $c^\star \ge Q_{weak}$ \textbf{then} \hfill \COMMENT{weak quorum $Q_{weak}=f{+}1$}
\STATE\quad $\langle \text{PoA},\; h^\star \rangle \gets \textsc{BuildCert}(r,\; h^\star,\; c^\star,\; \text{PoA})$
\STATE\quad append $\langle \text{PoA},\; h^\star \rangle$ to $\mathcal{L}_i$
\STATE\quad \textsc{UpdateHiCert}$(\langle \text{PoA},\; h^\star \rangle)$
\SPACE \COMMENT{PoA preserves progress, but does not justify next-round entry}
% \STATE\quad \textsc{ExitRound}$(r,\; \langle \text{PoA},\; h^\star \rangle)$
% \STATE \textbf{else}
\STATE \textsc{ExitRound}$(r,\; \perp)$
\ENDEVENT
\SPACE

\EVENT{\textsc{ExitRound}$(r,\; cert)$}
\STATE \textbf{if} $r \ne r_i$ \textbf{then} return
\STATE Cancel $\textit{ExchangeTimer}[r]$;\; Cancel $\textit{ElectionTimer}[r]$
\STATE Call \textsc{pacemaker.CompleteRound}$(r,\; cert)$
\ENDEVENT

\end{myprotocol}
\vspace{-2.5mm}
\caption{Pseudocode of Round Logic in \sysname}
\label{fig:round_logic}
\vspace{-4mm}
\end{figure}

\subsubsection{Locking and Commitment} 
\label{sec:commitment}

Replicas independently interpret their local logs and finalize proposals by applying a deterministic locking and commit rule over PoRs.

\vspace{1.5mm}
\iheading{Locking Rule.}
Each replica maintains a local variable $lock_i$ representing the highest PoR-certified proposal it has observed.
A replica votes only for proposals that extend $lock_i$, unless it observes a PoR from a strictly higher round.

\vspace{1.5mm}
\iheading{Commitment Rule.}
A proposal $b_r$ is committed once it forms a two-PoR chain.
Formally, if $b_r$ and its immediate successor $b_{r+1}$ both obtain PoRs (i.e., $\PoR{}_r$ and $\PoR{}_{r+1}$), and $b_{r+1}$ extends $b_r$, then $b_r$ and all its ancestors become committed (Figure~\ref{fig:commit}).

\vspace{1.5mm}
\iheading{Safety Intuition.}
Each \PoR{} is formed by a strong quorum of size $2f{+}1$.
Any two such quorums intersect in at least $f{+}1$ replicas, including at least one honest replica.
Once $\PoR{}_r$ and $\PoR{}_{r+1}$ are formed for a chain extending $b_r$, at least $f{+}1$ honest replicas participating in $\PoR{}_{r+1}$ become locked on the proposal containing $b_r$.
Due to quorum intersection and the single-vote rule, no conflicting proposal can subsequently gather a strong quorum in round $r{+}1$ or beyond.
Hence, two conflicting proposals cannot both satisfy the two-PoR commit condition, ensuring the safety of consensus.

\subsection{Round Synchronization}
\label{ssec:pacemaker}
% \Dakai{TODO: Add the intuitions behind the pacemaker design.}

To form either a \PoR{} or a \PoA{}, replicas must remain in the same round long enough to exchange proposals and votes. 
\new{
%Thus, a pacemaker in \sysname must serve two purposes:
%(i) keep replicas working on the same logical round, and
%(ii) ensure that replicas have sufficient overlap in that round to exchange messages and form those certificates.
Classical pacemakers often place this synchronization logic directly on the critical consensus path and 
require a strong quorum of replicas for synchronization. 
Under network partitions, such designs stall.
}

\new{
%Unlike classical pacemakers that repeatedly rely on strong-quorum synchronization, 
%\sysname advances rounds with certificates that justify that the current round has already terminated:
\sysname{} instead deploys a \emph{decoupled pacemaker}: 
round advancement remains certificate-driven on the critical path, 
while timeouts are calibrated off the critical path through a background synchronization service. 
This design simultaneously supports: 
(i) partial progress within weak connected components of at least $f{+}1$ replicas, and 
(ii) fast catch-up and convergence once $2f{+}1$ replicas become connected.
%and (iii) adaptive timeout calibration under unknown and time-varying
%network delay
}

%The pacemaker must precisely define when a replica is allowed to advance.
%\sysname{} answers this by restricting round advancement to two certificate-backed conditions.

\subsubsection{Round Advancement Conditions.} 
\label{subsubsec:round_advancement}
After the election and certification phase, replicas participate in the round advancement phase.
A replica in round $r$ advances to round $r{+}1$ in exactly one of the following two cases:
\begin{itemize}[nosep,wide]
    \item \textbf{PoR-based.}
    If a proposal in round $r$ gathers a strong quorum of votes and forms a $\PoR{}_r$, replicas advance to round $r{+}1$.

    \item \textbf{RC-based advancement.}
    If no proposal forms a $\PoR{}_r$ before the election timer expires, 
    the replica broadcasts a \textsc{WishNewRound}(r) message. 
    Once it collects $f{+}1$ such messages, it forms a \textit{Round Certificate} ($\RC{}_r$) and advances to round $r{+}1$.
\end{itemize}

\new{Crucially, these rules control logical round advancement in the consensus path: replicas continue progressing through proposal exchange, voting, and \PoR{}/\RC{} formation without waiting for a separate timeout synchronization procedure to complete.
}

\subsubsection{Quick Catch-Up and Fast Convergence.}
Once a replica forms a \PoR{} or \RC{} certificate, 
it broadcasts the certificate so that any lagging replica can immediately jump to the {\em justified} round. 
Concretely, if a replica is currently in round $r'$ and receives a valid \PoR{} or \RC{} formed in round $r > r'$, 
it advances to round $r{+}1$.

%A $\PoR{}_r$ certifies that a strong quorum has completed round $r$, while an $\RC{}_r$ certifies that at least $f{+}1$ replicas have timed out in round $r$ and requested advancement.
%In either case, the certificate provides justification that round $r$ has terminated and that round $r{+}1$ should begin.

Since round numbers increase monotonically and replicas only accept certificates for strictly higher rounds, certificate dissemination induces convergence toward the highest justified round.
If a lagging replica receives $\PoR{}_r$, it learns that round $r$ has successfully produced a strong certificate and advances accordingly.
If it receives $\RC{}_r$, it learns that an honest replica in round $r$ could not form a \PoR{} and synchronizes by advancing to round $r{+}1$.

During synchronous periods after GST, such certificates propagate within $\Delta$, bounding round divergence among correct replicas by at most one message delay.
This also explains why the \emph{Proposal Exchange} phase spans  $2\delta_i$: one $\delta_i$ for round synchronization and one  $\delta_i$ for proposal dissemination.
\new{In particular, once a strong quorum becomes connected again, certificate dissemination rapidly pulls lagging replicas toward the highest justified round, without requiring an explicit synchronization barrier on the critical path.
}
%\emph{Note:} Although \PoA{}s enable partial progress during network partitions,
%they alone \emph{cannot} justify round advancement.
%Allowing weak-quorum certificates to trigger advancement would let Byzantine replicas force premature round transitions, 
%preventing proposals from remaining in the same round long enough to form a \PoR{}.
%This would compromise eventual strong-quorum progress and thus undermine liveness.

\begin{figure}[t!]
    \vspace{-4mm}
    \centering
    \begin{tikzpicture}[
        proposal/.style = {
          rectangle, draw=black, rounded corners=2pt,
          minimum width=1cm, minimum height=0.5cm,
          align=center, font=\footnotesize
        },
        por/.style  = {proposal, fill=blue!20},
        poa/.style  = {proposal, fill=green!20},
        pop/.style  = {proposal, fill=gray!10},
        label/.style = {draw, fill=yellow!20, rounded corners=1pt, font=\scriptsize, align=center},
        label_skip/.style = {draw, fill=red!20, rounded corners=1pt, font=\scriptsize, align=center}
      ]

      \node[por] (p1) {Proposal 1\\PoR};
      \node[pop] (p2) [right=0.5mm of p1] {Round 2\\RC};
      \node[poa] (p3) [right=0.5mm of p2] {Proposal 3\\PoA};
      \node[poa] (p4) [right=0.5mm of p3] {Proposal 4\\PoA};
      \node[por] (p5) [right=0.5mm of p4] {Proposal 5\\PoR};

      \node[label_skip] (l1) [above=2.5mm of p1] {Not ready\\to commit};
      \draw[-{Stealth[scale=0.7]}, thick] (l1.south) -- (p1.north);
      \node[label_skip] (l2) [above=2.5mm of p2] {Not ready\\to commit};
      \draw[-{Stealth[scale=0.7]}, thick] (l2.south) -- (p2.north);
      \node[label_skip] (l3) [above=2.5mm of p3] {Not ready\\to commit};
      \draw[-{Stealth[scale=0.7]}, thick] (l3.south) -- (p3.north);

      \node[label_skip] (l4) [above=2.5mm of p4] {Not ready\\to commit};
      \draw[-{Stealth[scale=0.7]}, thick] (l4.south) -- (p4.north);

      \node[label_skip] (l5) [above=2.5mm of p5] {Not ready\\to commit};
      \draw[-{Stealth[scale=0.7]}, thick] (l5.south) -- (p5.north);

       \draw[thick, decorate, decoration={snake,amplitude=.4mm,segment length=2mm}, ->, line width=1pt] ($(p3.south)-(0,1ex)$) -- ++(0,-4mm) node[midway, right=0.5em, font=\small, align=left] {A new PoR generated in round 6};

      \node[por] (p21) [below=15mm of p1] {Proposal 1\\PoR};
      \node[pop] (p22) [right=0.5mm of p21] {Round 2\\RC};
      \node[poa] (p23) [right=0.5mm of p22] {Proposal 3\\PoA};
      \node[poa] (p24) [right=0.5mm of p23] {Proposal 4\\PoA};
      \node[por] (p25) [right=0.5mm of p24] {Proposal 5\\PoR};
      \node[por] (p26) [right=0.5mm of p25] {Proposal 6\\PoR};

      \node[label] (l21) [above=2.5mm of p21] {Ready to\\commit};
      \draw[-{Stealth[scale=0.7]}, thick] (l21.south) -- (p21.north);
      \node[label, fill=gray!10] (l22) [above=2.5mm of p22] {Skip\\this round};
      \draw[-{Stealth[scale=0.7]}, thick] (l22.south) -- (p22.north);
      \node[label] (l23) [above=2.5mm of p23] {Ready to\\commit};
      \draw[-{Stealth[scale=0.7]}, thick] (l23.south) -- (p23.north);

      \node[label] (l24) [above=2.5mm of p24] {Ready to\\commit};
      \draw[-{Stealth[scale=0.7]}, thick] (l24.south) -- (p24.north);

      \node[label] (l25) [above=2.5mm of p25] {Ready to\\commit};
      \draw[-{Stealth[scale=0.7]}, thick] (l25.south) -- (p25.north);

      \node[label_skip] (l26) [above=2.5mm of p26] {Not ready\\to commit};
      \draw[-{Stealth[scale=0.7]}, thick] (l26.south) -- (p26.north);
    \end{tikzpicture}

\vspace{-3mm}
\caption{An Example of Two-PoR Commit Rule}
\label{fig:commit}
\vspace{-4mm}
\end{figure}

\subsubsection{Background Timeout Calibration} \label{subsubsec:background-resync}

\new{
\sysname{} operates under a partially synchronous network model, in which message delays are bounded by some $\Delta$
that holds only during stabilized periods (after some GST) and whose value is not known a priori. 
Because the actual delay can differ across stable periods---particularly after long partitions or asymmetric delays---replicas' local timeouts may become misaligned. 
To address this, existing pacemakers perform timeout calibration, but they do so on the critical path.
}

\begin{figure}[tp]
\vspace{-2mm}
\begin{myprotocol}
\TITLE{Pacemaker$_i$}{Round Advancement \& Timeout Calibration}
\SPACE\textbf{Local state}: 
 \new{\SPACE\quad $r_i,\;K$ \hfill \COMMENT{current round and timeout calibration interval}
\SPACE\quad $\delta_i,\;\nu_i$ \hfill \COMMENT{local timeout baseline and current sync view}
\SPACE\quad $\textit{Calibrating}_i$ \hfill \COMMENT{whether a calibration attempt is in progress}
\SPACE\quad $ReadySet[\nu],\;CertSet[\nu]$ \hfill \COMMENT{valid $\langle SYNC \rangle$ messages}
\SPACE
}\FUNCTION{\textsc{CompleteRound}}{$r,\; cert$}
\STATE \textbf{if} $r\;mod\;K=0$ \textbf{then} \hfill \COMMENT{Reach the \textsc{SyncTimeout} boundary.}
\STATE\quad Call \textsc{pacemaker.SyncTimeout}
\STATE \textbf{if} $cert$ is \PoR{} \textbf{then}
\STATE\quad Call \textsc{EnterRound}$(r{+}1,\; cert,\; 2\delta_i)$
\SPACE\quad \COMMENT{A PoR directly justifies entry into the next round}
\STATE \textbf{else}  Broadcast $\langle \textsc{WishNewRound}(r) \rangle$
\new{\SPACE\quad \COMMENT{If no PoR, request for an RC for round advancement}
}\ENDFUNCTION
\SPACE

\EVENT{Received $f{+}1$ $\langle \textsc{WishNewRound}(r) \rangle$ messages}
\STATE $\langle RC_r \rangle \gets \textsc{AggregateSig}(f{+}1\ \langle \textsc{WishNewRound}(r) \rangle)$
\STATE Broadcast($\langle RC_r \rangle$)
\ENDEVENT
\SPACE

\EVENT{Received $\langle RC_r \rangle$}
\STATE Call \textsc{EnterRound}$(r{+}1,\; RC_r,\; 2\delta_i)$
\SPACE \COMMENT{An $RC_r$ safely coordinates entry into round $r{+}1$}
\ENDEVENT
\SPACE

\FUNCTION{\textsc{ObserveHigherRound}}{$P$}
\STATE \textbf{if} $P.r \le r_i$ \textbf{or} \textsc{ValidEnterCert}$(P.enterCert,\; P.r)$ = False \textbf{then} return \hfill \COMMENT{The $enterCert$ certifies the round entry}
\SPACE \COMMENT{Use a shortened timeout since round $r'$ is already in progress}
\STATE Call \textsc{EnterRound}$(P.r,\; P.enterCert,\; \delta_i)$
\ENDFUNCTION

\new{\SPACE
\EVENT{\textsc{SyncTimeout}}
\STATE \textbf{if} $\textit{Calibrating}_i$ = False \textbf{then}
\STATE\quad $\nu_i \gets \nu_i + 1;\;\; \textit{Calibrating}_i \gets$ True
\STATE\quad Broadcast($\langle \mathsf{SYNC\mbox{-}READY},\; \nu_i \rangle_i$)
\ENDEVENT
\SPACE
}

\new{\EVENT{Received valid $\langle \mathsf{SYNC\mbox{-}READY},\; \nu \rangle$ \textbf{from} $p_j$}
\STATE $ReadySet[\nu][j] \gets\ \langle \mathsf{SYNC\mbox{-}READY},\; \nu \rangle_j$
\STATE $V^{hi} \gets \{\nu' \mid \nu' > \nu_i \text{ and } |ReadySet[\nu']| \ge f+1\}$
\STATE \textbf{if} $V^{hi} \neq \emptyset$ \textbf{then}
\STATE\quad $\nu_i \gets \min V^{hi}$; $\textit{Calibrating}_i \gets$ True
\STATE\quad Broadcast($\langle \mathsf{SYNC\mbox{-}READY},\; \nu_i \rangle_i$)
\STATE \textbf{if} $\textit{CertSet}_i[\nu][i] =\, \perp$ \AND $|ReadySet[\nu]| \ge 2f+1$ \textbf{then}
\STATE\quad $\textit{CertSet}_i[\nu][i] \gets\,\langle \mathsf{SYNC\mbox{-}CERT},\; \nu \rangle_i$
\STATE\quad Broadcast($\langle \mathsf{SYNC\mbox{-}CERT},\; \nu \rangle_i$)
\STATE\quad \textbf{start} $SyncTimer[\nu]$ \textbf{with timeout} $\delta_i$
\ENDEVENT
\SPACE
}

\new{\EVENT{Received valid $\langle \mathsf{SYNC\mbox{-}CERT},\; \nu \rangle$ \textbf{from} $p_j$}
\STATE $CertSet[\nu][j] \gets \langle \mathsf{SYNC\mbox{-}CERT},\; \nu \rangle_j$
\STATE \textbf{if} $CertSet[\nu][i]\,\neq\,\perp$ \AND $|CertSet[\nu]| > f+1$ \textbf{then}
\STATE\quad $\textit{Calibrating}_i \gets$ False;  Cancel $SyncTimer[\nu]$
\SPACE\quad \COMMENT{$\delta_i$ is sufficient; decrease if success was fast (\Cref{subsubsec:background-resync})}
\ENDEVENT
\SPACE
}

\new{\EVENT{$SyncTimer[\nu]$ \textbf{timeout}}
\STATE \textbf{if} $|CertSet[\nu]| \le f$ \textbf{then}
\STATE\quad $\delta_i \gets 2\delta_i$; $\nu_i \gets \nu_i + 1$
\STATE\quad Broadcast($\langle \mathsf{SYNC\mbox{-}READY},\; \nu_i \rangle_i$)
\ENDEVENT
}

\end{myprotocol}
\vspace{-2.5mm}
\caption{Pseudocode of the Pacemaker}
\label{fig:pacemaker}
% Keep the caption clear of the following float in the extended layout.
\end{figure}

\new{To address this, \sysname runs a periodic background timeout-calibration service, denoted $\mathsf{SyncTimeout}$, which adjusts each replica’s local timeout value ($\delta_i$).
This service is \textit{off the critical path}: it never blocks the chain construction process.
Replicas initiate $\mathsf{SyncTimeout}$ periodically, e.g., once every $K$ rounds ($K$ is a system parameter), provided that no prior calibration attempt is still in progress.
}

\new{Each calibration attempt is identified by a synchronization view number
$\nu$.
The view number is used to distinguish retries and prevent replay of stale synchronization messages.
If a calibration attempt fails, replicas keep retrying in the background using higher synchronization views, while the main protocol remains unchanged and continues running.
}

\vspace{1mm}
\new{\iheading{Phase 1: \texorpdfstring{$\langle \mathsf{SYNC\mbox{-}READY} \rangle$}{SYNC-READY}.}
After every $K$-th round, if replica $p_i$ is not undergoing a timeout-calibration, 
it initiates a calibration attempt in synchronization view $\nu$ by broadcasting
\(
\langle \mathsf{SYNC\mbox{-}READY}, \nu \rangle_i.
\)
Replica $p_i$ keeps periodically retransmitting this message until it receives either 
$2f{+}1$ such messages for view $\nu$ or $f{+}1$ such messages for higher views $\nu'$.
%$\langle \mathsf{SYNC\mbox{-}READY}, \nu \rangle$ messages or $f+1$ $\langle \mathsf{SYNC\mbox{-}READY}, \nu' \rangle$ messages from higher views.
In the latter case, replica $p_i$ abandons view $\nu$ and joins the smallest view $\nu'$ among these 
$f{+}1$ messages by broadcasting
\(
\langle \mathsf{SYNC\mbox{-}READY}, \nu' \rangle_i.
\)
}

\vspace{1mm}
\new{\iheading{Phase 2: \texorpdfstring{$\langle \mathsf{SYNC\mbox{-}CERT} \rangle$}{SYNC-CERT}.}
Once a replica $p_i$ collects $2f{+}1$ valid $\mathsf{SYNC}$ $\mathsf{\mbox{-}READY}$ messages for its synchronization view $\nu$, 
it creates and broadcasts a
\(
\langle \mathsf{SYNC\mbox{-}CERT}, \nu \rangle_i
\)
message.
After that, $p_i$ waits on a timer of length $\delta_i$ for $\mathsf{SYNC\mbox{-}CERT}$ messages from other replicas.
}

\new{\iheading{Success Condition.}
The calibration attempt in view $\nu$ succeeds at replica $p_i$ if:
$p_i$ formed a $\langle \mathsf{SYNC\mbox{-}CERT}, \nu \rangle_i$, and
within period $\delta_i$, it received $\langle \mathsf{SYNC\mbox{-}CERT}, \nu \rangle$ from at least $f+1$ distinct replicas (i.e., the total number of certificates, including its own, exceeds $f{+}1$).
}

\new{This condition ensures that replica $p_i$'s current timeout $\delta_i$ is long enough to allow at least one honest replica to generate and send back a certificate under the current network condition.
In this case, replica $p_i$ keeps its current timeout baseline unchanged.
}

\new{\iheading{Failure handling.}
We distinguish two failure cases:
}\begin{enumerate}[nosep,wide]
    \item \new{\textbf{No local certificate formed.}
    If replica $p_i$ cannot collect $2f{+}1$ valid $\mathsf{SYNC\mbox{-}READY}$ messages for view $\nu$, it leaves $\delta_i$ unchanged and keeps retransmitting
    $\langle \mathsf{SYNC\mbox{-}READY}, \nu \rangle_i$.
}

    \item \new{\textbf{Insufficient remote certificates received.}
    If replica $p_i$ succeeds in forming $\langle \mathsf{SYNC\mbox{-}CERT}, \nu \rangle_i$ but does not receive this certificate from at least $f{+}1$ distinct replicas within time $\delta_i$, the attempt is treated as a failure.
    In this case, replica $p_i$ doubles its local timeout period, i.e.,  $\delta_i \leftarrow 2\delta_i$, advances to the next synchronization view $\nu \leftarrow \nu + 1$, and restarts the procedure by broadcasting a new $\langle \mathsf{SYNC\mbox{-}READY}, \nu \rangle$.
}\end{enumerate}

%\new{Overall, $\mathsf{SyncTimeout}$ is a calibration mechanism in the decoupled pacemaker:
%logical rounds still advance only through \PoR{} and \RC{}, while timeout calibration runs periodically in the background and never blocks the normal consensus path.
%}
\begin{figure}[t!]
    \centering
    \resizebox{\linewidth}{!}{%
    \begin{tikzpicture}[>=stealth, thick]
    \useasboundingbox (-0.5,0.7) rectangle (9.0,-1.9);

    \draw (0,0) -- (8.1,0);
    \node[left] at (0,0) {R1};
    \draw (0,-0.6) -- (8.1,-0.6);
    \node[left] at (0,-0.6) {R2};
    \draw (0,-1.2) -- (8.1,-1.2);
    \node[left] at (0,-1.2) {R3};
    \draw[color=red] (0,-1.8) -- (8.1,-1.8);
    \node[left, color=red] at (0,-1.8) {R4};

    \node[align=center, draw, fill=white, rounded corners=1pt, inner sep=2pt] at (1.5,0) {$\langle SYNC\mbox{-}READY \rangle$};
    \node[align=center, draw, fill=white, rounded corners=1pt, inner sep=2pt] at (1.5,-0.6) {$\langle SYNC\mbox{-}READY \rangle$};
    \node[align=center, draw, fill=white, rounded corners=1pt, inner sep=2pt] at (1.5,-1.2) {$\langle SYNC\mbox{-}READY \rangle$};
    % \node[align=center, draw=red, text=red, fill=white, rounded corners=1pt, inner sep=2pt] at (1,-1.8) {$\langle SYNC\mbox{-}READY \rangle$};

    \draw[<->, blue, thick] (0.1,0.3) -- (4.1,0.3) node[midway, above, text height=1.5ex, text depth=0.25ex] {$\langle SYNC\mbox{-}READY \rangle$};

    \draw[->] (2.7,0) -- (4.0,-0.6);
    \draw[->] (2.7,0) -- (4.0,-1.2);
    \draw[->] (2.7,0) -- (4.0,-1.8);
    \draw[->] (2.7,-0.6) -- (4.0,0);
    \draw[->] (2.7,-0.6) -- (4.0,-1.2);
    \draw[->] (2.7,-0.6) -- (4.0,-1.8);
    \draw[->] (2.7,-1.2) -- (4.0,0);
    \draw[->] (2.7,-1.2) -- (4.0,-0.6);
    \draw[->] (2.7,-1.2) -- (4.0,-1.8);

    \draw[dashed] (4.1, 0.25) -- (4.1, -2.0);

    \draw[<->, blue, thick] (4.1,0.3) -- (7.95,0.3) node[midway, above, text height=1.5ex, text depth=0.25ex] {$\langle SYNC-CERT \rangle$};
    \node[align=center, draw, fill=white, rounded corners=1pt, inner sep=2pt] at (5.3,0) {$\langle SYNC\mbox{-}CERT \rangle$};
    \node[align=center, draw, fill=white, rounded corners=1pt, inner sep=2pt] at (5.3,-0.6) {$\langle SYNC\mbox{-}CERT \rangle$};
    \node[align=center, draw, fill=white, rounded corners=1pt, inner sep=2pt] at (5.3,-1.2) {$\langle SYNC\mbox{-}CERT \rangle$};
    \node[align=center, draw, fill=white, rounded corners=1pt, inner sep=2pt] at (5.3,-1.8) {$\langle SYNC\mbox{-}CERT \rangle$};

    \draw[->] (6.35,0) -- (7.65, -0.6);
    \draw[->] (6.35,0) -- (7.65, -1.2);
    \draw[->] (6.35,0) -- (7.65, -1.8);
    \draw[->] (6.35,-0.6) -- (7.65, 0);
    \draw[->] (6.35,-0.6) -- (7.65, -1.2);
    \draw[->] (6.35,-0.6) -- (7.65, -1.8);
    \draw[->] (6.35,-1.2) -- (7.65, 0);
    \draw[->] (6.35,-1.2) -- (7.65, -0.6);
    \draw[->] (6.35,-1.2) -- (7.65, -1.8);

    \end{tikzpicture}%
    }
    
    \vspace{-2mm}
    \newcaption{Timeout Calibration}
    \label{fig:timeout_calibration}
    \vspace{-4mm}
\end{figure}

\new{
\paragraph{Timeout Decrease.}
Classical BFT protocols such as PBFT~\cite{pbftj} double the view-change timeout whenever a view change fails to make progress,
but never decrease it. 
As a result, a single transient period of high delay can leave rounds unnecessarily long well after the delay has subsided. 
\sysname{} optionally permits decreasing $\delta_i$ upon a sufficiently fast calibration success.
Concretely, if a calibration attempt succeeds and the elapsed time since replica $p_i$ formed its local $\langle \mathsf{SYNC\mbox{-}CERT}, \nu \rangle_i$ is less than $\delta_i / \alpha$ for a pre-set threshold $\alpha > 1$ (e.g., $\alpha = 4$), $p_i$ halves its local timeout baseline: $\delta_i \leftarrow \max(\delta_i / 2,\; \delta_{\min})$, where $\delta_{\min}$ is a pre-configured lower bound.
The subsequent calibration attempt then verifies whether the reduced timeout is still sufficient; if not, it is doubled back.
This yields a multiplicative-increase/multiplicative-decrease pattern analogous to TCP congestion control~\cite{jacobson1988congestion}, adapted to the BFT setting where safety is never at risk because $\delta_i$ affects only round duration, not the correctness of voting or certification.
}

%\new{\iheading{Retransmission.}
%To prevent the calibration mechanism from stalling when a replica cannot collect $2f{+}1$ $\mathsf{SYNC\mbox{-}READY}$ messages (e.g., during a prolonged partition or view misalignment), replicas periodically retransmit their current $\langle \mathsf{SYNC\mbox{-}READY}, \nu_i \rangle_i$ message until the ongoing calibration attempt concludes.
%This ensures that once connectivity is restored, replicas eventually converge on a common synchronization view and complete the calibration.
%}

\begin{figure}[t!]
\vspace{-2mm}
\begin{myprotocol}
\new{\TITLE{Pacemaker (Timeout Calibration)}{replica $p_i$}
}

\new{\SPACE\textbf{Local State}
\SPACE\quad $\delta_i,\;\nu_i$ \hfill \COMMENT{local timeout baseline and current sync view}
\SPACE\quad $\textit{Calibrating}_i$ \hfill \COMMENT{whether a calibration attempt is in progress}
\SPACE\quad $ReadySet[\nu]$ \hfill \COMMENT{valid $\langle \mathsf{SYNC\mbox{-}READY}, \nu \rangle$ messages}
\SPACE\quad $CertSet[\nu]$ \hfill \COMMENT{valid $\langle \mathsf{SYNC\mbox{-}CERT}, \nu \rangle$ messages}
\SPACE
}

\new{\EVENT{SyncTimeout}
\STATE \textbf{if} $\textit{Calibrating}_i$ = False \textbf{then}
\STATE\quad $\nu_i \gets \nu_i + 1;\;\; \textit{Calibrating}_i \gets$ True
\STATE\quad Broadcast($\langle \mathsf{SYNC\mbox{-}READY},\; \nu_i \rangle_i$)
\ENDEVENT
\SPACE
}

\new{\EVENT{Received valid $\langle \mathsf{SYNC\mbox{-}READY},\; \nu \rangle$ \textbf{from} $p_j$}
\STATE $ReadySet[\nu][j] \gets\ \langle \mathsf{SYNC\mbox{-}READY},\; \nu \rangle_j$
\STATE $V^{hi} \gets \{\nu' \mid \nu' > \nu_i \text{ and } |ReadySet[\nu']| \ge f+1\}$
\STATE \textbf{if} $V^{hi} \neq \emptyset$ \textbf{then}
\STATE\quad $\nu_i \gets \min V^{hi}$; $\textit{Calibrating}_i \gets$ True
\STATE\quad Broadcast($\langle \mathsf{SYNC\mbox{-}READY},\; \nu_i \rangle_i$)
\STATE \textbf{if} $\textit{CertSet}_i[\nu][i] =\, \perp$ \AND $|ReadySet[\nu]| \ge 2f+1$ \textbf{then}
\STATE\quad $\textit{CertSet}_i[\nu][i] \gets\,\langle \mathsf{SYNC\mbox{-}CERT},\; \nu \rangle_i$
\STATE\quad Broadcast($\langle \mathsf{SYNC\mbox{-}CERT},\; \nu \rangle_i$)
\STATE\quad \textbf{start} $SyncTimer[\nu]$ \textbf{with timeout} $\delta_i$
\ENDEVENT
\SPACE
}

\new{\EVENT{Received valid $\langle \mathsf{SYNC\mbox{-}CERT},\; \nu \rangle$ \textbf{from} $p_j$}
\STATE $CertSet[\nu][j] \gets \langle \mathsf{SYNC\mbox{-}CERT},\; \nu \rangle_j$
\STATE \textbf{if} $CertSet[\nu][i]\,\neq\,\perp$ \AND $|CertSet[\nu]| > f+1$ \textbf{then}
\STATE\quad $\textit{Calibrating}_i \gets$ False;  Cancel $SyncTimer[\nu]$
\SPACE\quad \COMMENT{$\delta_i$ is sufficient; decrease if success was fast (\Cref{subsubsec:background-resync})}
\ENDEVENT
\SPACE
}

\new{\EVENT{$SyncTimer[\nu]$ \textbf{timeout}}
\STATE \textbf{if} $|CertSet[\nu]| \le f$ \textbf{then}
\STATE\quad $\delta_i \gets 2\delta_i$; $\nu_i \gets \nu_i + 1$
\STATE\quad Broadcast($\langle \mathsf{SYNC\mbox{-}READY},\; \nu_i \rangle_i$)
\ENDEVENT
}

\end{myprotocol}

\vspace{-2.5mm}
\newcaption{Pseudocode of Background Timeout Calibration}
\label{fig:timeout_calibration_protocol}
\vspace{-3mm}
\end{figure}

\subsection{Merge after Network Recovery}
A \PoA{}-backed branch guarantees data availability, provides ordering evidence, and reduces wasted work after recovery. 
However, it is provisional and does not imply final commitment. 
During a network partition, different connected components may therefore extend the same \PoR{} along divergent \PoA{} chains.
\new{The same reasoning applies to asymmetric partitions in which a Byzantine replica selectively forwards proposals across components: forwarding may create additional provisional \PoA{} branches, which introduces additional computation costs linear in the number of PoAs for verification, but does not force a pairwise merge of all branches or roll back any committed state.}
%\new{Although these \PoA{}-backed branches are provisional, they preserve recoverable availability and ordering evidence, reducing lost work after recovery.}

When network connectivity is restored, replicas resume exchanging proposals in the {\em proposal exchange} phase. 
Each proposal carries its sender's highest known \PoR{} and \PoA{}, 
summarizing the strongest certificates that replica has observed. 
During the election phase, replicas then apply the deterministic proposal-priority rule: 
proposals anchored at a higher-round \PoR{} are always preferred, and ties are broken by the highest-round \PoA{} and 
then the threshold coin. 

Consequently, once proposals from previously disconnected components become visible across the network, 
correct replicas deterministically converge on the same strongest branch and vote for it. 
This allows a new \PoR{} to form on that branch, re-establishing a single, globally agreed \PoR{}.
The remaining \PoA{} branches---those not extended by a subsequent \PoR{}---stop growing and are eventually superseded. 
Because \PoA{} formation guarantees data availability, the proposals on these branches remain retrievable; 
however, only the branch incorporated into a subsequent \PoR{} can be finalized. 

Merging in \sysname{} is therefore implicit and protocol-driven: 
no separate reconciliation procedure is required, as the normal proposal-selection and 
certification process automatically converges the system back to a single chain once connectivity is restored.

%When network connectivity is restored, replicas resume exchanging proposals through the {\em proposal exchange} phase.
%Each proposal carries its highest known \PoR{} and \PoA{}, summarizing the strongest certificates observed by the sender.
%Replicas then apply the deterministic proposal-priority rule during the election phase: proposals anchored at a higher-round \PoR{} are always preferred, and ties are broken using the highest-round \PoA{} and the threshold coin.
%
%As a result, once proposals from previously disconnected components become visible across the network, correct replicas deterministically converge on the same strongest branch and vote for it.
%This enables the formation of a new \PoR{} extending that branch, thereby re-establishing a single globally known \PoR{}.
%%
%Other \PoA{} branches that are not extended by a subsequent \PoR{} naturally stop growing and are eventually superseded.
%Since \PoA{} formation guarantees data availability, the proposals on those branches remain retrievable, but only the branch incorporated into a subsequent \PoR{} can be finalized.
%
%Therefore, merging in \sysname{} is implicit and protocol-driven.
%No separate reconciliation procedure is required: once connectivity is restored, the normal proposal selection and certification process automatically converges the system back to a single chain.
%

\section{\MakeUppercase{\sysname}: DUAL-PATH DESIGN}

 \new{
 Although \sysname{} achieves partial progress under network partitions, it relies on all-to-all proposal exchange and conservative timeouts; 
 as a result, it incurs \emph{quadratic communication complexity} and forgoes \emph{optimistic responsiveness}~\cite{hotstuff}, 
 preventing replicas from proceeding at network speed even when the network is stable. 
 In this section, we present an optimized version of \sysname{} based on a \emph{dual-path design}. 
 In the absence of leader or network failures, 
 the system runs an optimistic \emph{fast path} that assumes a \emph{designated} leader and 
 lets replicas reach consensus with linear communication. 
%Further, the fast path reduces latency through its optimistic responsiveness design. 
 If the \emph{fast path} fails, replicas fall back to the \emph{base path} (\S\ref{sec:overview}).
}

 \subsection{Linear-Complexity Fast Path}
\sysname's fast-path assumes a dedicated leader for each round; 
let $L_r$ be the leader for round $r$ (leaders are selected in a round-robin manner). 
The fast-path undergoes the following steps (Fig.~\ref{fig:fast_path}).

\begin{itemize}[nosep,wide]
    \item \textbf{Fast-path proposal.}
    On entering round $r$ (i.e., on receiving a $\PoR_{r-1}$ or an $\RC_{r-1}$), 
    the leader $L_r$ broadcasts a proposal $b_r$.
    This proposal extends the highest known branch and carries the sender's latest $\widehat{PoR}$ and $\widehat{PoA}$ certificates
    (similar to the base path).
     \item \textbf{Linear vote collection.}
    Upon receiving $b_r$, each replica verifies that it is valid and lock-compatible (cf.~\Cref{sec:commitment}).
    If this is the case, the replica sends its vote share to a designated collector, e.g., the next-round leader $L_{r+1}$.
    %This preserves the same voting rule while replacing all-to-all vote exchange with a linear collection pattern.

    \item \textbf{Pipelined certificate formation.}
    Once $L_{r+1}$ collects a strong quorum of vote shares for $b_r$, it aggregates them into a $\PoR_r$.
    This $\PoR_r$ then serves as the justification to enter round $r{+}1$, allowing $L_{r+1}$ to immediately propose the next block. %and pipeline the next round.
\end{itemize}

In summary, the fast path modifies only the \textit{message-collection pattern} used to form a \PoR{}; 
it does not alter the safety rule, the locking rule, or the certificate semantics. 
As a result, under stable conditions, \sysname{} substantially reduces its communication complexity, from quadratic to linear.
%dependence on the conservative waiting required by the decentralized proposal-exchange phase, and progress can approach the pace of actual network latency in the common case.

\begin{figure}[t]
    \vspace{0mm}
    \centering
    \resizebox{\linewidth}{!}{%
    \begin{tikzpicture}[>=stealth, thick]
    \useasboundingbox (-0.5,0.7) rectangle (8.5,-1.9);

    \draw (0,0) -- (8.2,0);
    \node[left] at (0,0) {\textbf{R1}};
    \draw (0,-0.6) -- (8.2,-0.6);
    \node[left] at (0,-0.6) {R2};
    \draw (0,-1.2) -- (8.2,-1.2);
    \node[left] at (0,-1.2) {R3};
    \draw[color=red] (0,-1.8) -- (8.2,-1.8);
    \node[left, color=red] at (0,-1.8) {R4};

    \node[align=center, draw, fill=white, rounded corners=1pt, inner sep=2pt] at (0.5,0) {P1$_r$};

    \node[align=center, draw, fill=white, rounded corners=1pt, inner sep=2pt] at (3.3,0) {Vote P1$_r$};
    \node[align=center, draw, fill=white, rounded corners=1pt, inner sep=2pt] at (3.3,-0.6) {Vote P1$_r$};
    \node[align=center, draw, fill=white, rounded corners=1pt, inner sep=2pt] at (3.3,-1.2) {Vote P1$_r$};

    \node[align=center, draw, fill=white, rounded corners=1pt, inner sep=2pt] at (6.3,-0.6) {PoR P1$_r$};
    \node[align=center, draw, fill=white, rounded corners=1pt, inner sep=2pt] at (7.5,-0.6) {P2$_{r+1}$};

    \draw[<->, blue, thick] (0.1,0.3) -- (0.9,0.3) node[midway, above] {\small \shortstack{Generate\\Proposal}};
    \draw[<->, blue, thick] (0.9,0.3) -- (2.7,0.3) node[midway, above] {\small \shortstack{Proposal\\Disseminate}};
    \draw[<->, blue, thick] (2.7,0.3) -- (5.8,0.3) node[midway, above] {\small \shortstack{Vote Collection\\by $L_{r+1}$}};
    \draw[<->, blue, thick] (5.8,0.3) -- (7.0,0.3) node[midway, above] {\small \shortstack{PoR\\Aggregation}};
    \draw[<-, blue, thick] (7.0,0.3) -- (8.23,0.3) node[midway, above] {\small \shortstack{Next\\Round}};

    \draw[->] (0.9,0) -- (2.65,-0.6);
    \draw[->] (0.9,0) -- (2.65,-1.2);
    \draw[->] (0.9,0) -- (2.65,-1.8);

    \draw[->] (3.9,0) -- (5.7,-0.6);
    \draw[->] (3.9,-1.2) -- (5.7,-0.6);

    \end{tikzpicture}%
    }
    \vspace{-6.5mm}
    \caption{Fast Path Workflow in Good Cases}
    \label{fig:fast_path}
    \vspace{-3mm}
\end{figure}

 \subsection{Fallback to the Base Path}
Each replica maintains a dedicated \textit{Fast-Path Timer} to detect failures. 
For example, if a replica does not receive a valid fast-path proposal from $L_r$ in time, or 
if the designated collector fails to gather a strong quorum of vote shares before the timer expires, 
the replica abandons the fast path and falls back to the \emph{base path} (Fig.~\ref{fig:fallback_base_path}).

The base path resumes within the same round and follows the steps described in \Cref{subsubsec:chain_growth}.
Concretely, each replica broadcasts its own round-$r$ proposal and participates in the 
{\em proposal exchange} and {\em election and certification} phases.
These proposals include locally known highest $\widehat{PoR}$ and $\widehat{PoA}$ certificates, 
so safety and proposal priority remain unchanged.

On entering the next round, replicas again attempt the fast path. 
The transition between the fast and base paths is therefore seamless: 
the fast path is a performance optimization layered on the base protocol, not a separate consensus protocol.
%Whenever the optimistic assumptions do not hold, replicas simply return to the existing partition-robust mechanism without affecting safety or eventual progress.

\subsection{\new{Dissemination Layer in \sysname}}
\new{
Prior work such as Narwhal~\cite{narwhal-tusk} and AutoBahn~\cite{autobahn} introduce the notion of a 
\emph{dissemination layer} that decouples data dissemination from ordering. 
Concretely, the dissemination layer continuously broadcasts batches of transactions and 
certifies their availability across replicas, 
while the ordering layer concurrently proposes references to the already-disseminated data. 
This decoupling is complementary to \sysname{}'s design and can be integrated with modest changes.
%
%Without this layer, unselected proposals may be discarded after each round, so \sysname{} commits at most one proposal per round; we refer to this baseline as \sysname{}-Raw.
}

\begin{figure}[t!]
    \vspace{-3mm}
    \centering
    \resizebox{\linewidth}{!}{%
     \begin{tikzpicture}[>=stealth, thick]
    \useasboundingbox (-0.5,0.7) rectangle (9.0,-2.5);

    \draw (0,0) -- (8.5,0);
    \node[left] at (0,0) {R1};
    \draw (0,-0.6) -- (8.5,-0.6);
    \node[left] at (0,-0.6) {R2};
    \draw (0,-1.2) -- (8.5,-1.2);
    \node[left] at (0,-1.2) {R3};
    \draw[color=red] (0,-1.8) -- (8.5,-1.8);
    \node[left, color=red] at (0,-1.8) {R4};

    \node[align=center, draw=red, text=red, fill=white, rounded corners=1pt, inner sep=2pt] at (0.4,-1.8) {P4$_r$};

    \draw[->, color=red, dashed] (0.7,-1.8) -- (1.7,0);
    \draw[->, color=red, dashed] (0.7,-1.8) -- (1.7,-0.6);
    \draw[->, color=red, dashed] (0.7,-1.8) -- (1.7,-1.2);
    \draw[dashed] (1.8, 0.25) -- (1.8, -2.0);
    \node[align=center] at (2.0, -2.35) {\textit{\shortstack{Fast-Path Timeout\\Fallback to Base Path}}};
    \node[align=center, fill=white, text=red] at (0.7, -0.5) {\shortstack{Leader\\Failed}};

    \node[align=center, draw, fill=white, rounded corners=1pt, inner sep=2pt] at (2.2,0) {P1$_r$};
    \node[align=center, draw, fill=white, rounded corners=1pt, inner sep=2pt] at (2.2,-0.6) {P2$_r$};
    \node[align=center, draw, fill=white, rounded corners=1pt, inner sep=2pt] at (2.2,-1.2) {P3$_r$};
    % \node[align=center, draw=red, text=red, fill=white, rounded corners=1pt, inner sep=2pt] at (3.65,-1.8) {Vote for P4};

    \draw[->] (2.5,0) -- (3.7,-0.6);
    \draw[->] (2.5,0) -- (3.7,-1.2);
    \draw[->] (2.5,0) -- (3.7,-1.8);
    \draw[->] (2.5,-0.6) -- (3.7,0);
    \draw[->] (2.5,-0.6) -- (3.7,-1.2);
    \draw[->] (2.5,-0.6) -- (3.7,-1.8);
    \draw[->] (2.5,-1.2) -- (3.7,0);
    \draw[->] (2.5,-1.2) -- (3.7,-0.6);
    \draw[->] (2.5,-1.2) -- (3.7,-1.8);

    \node[align=center, draw, fill=white, rounded corners=1pt, inner sep=2pt] at (4.35,0) {Vote P2$_r$};
    \node[align=center, draw, fill=white, rounded corners=1pt, inner sep=2pt] at (4.35,-0.6) {Vote P2$_r$};
    \node[align=center, draw, fill=white, rounded corners=1pt, inner sep=2pt] at (4.35,-1.2) {Vote P2$_r$};

    \draw[->] (5,0) -- (6.15,-0.6);
    \draw[->] (5,0) -- (6.15,-1.2);
    \draw[->] (5,0) -- (6.15,-1.8);
    \draw[->] (5,-0.6) -- (6.15,0);
    \draw[->] (5,-0.6) -- (6.15,-1.2);
    \draw[->] (5,-0.6) -- (6.15,-1.8);
    \draw[->] (5,-1.2) -- (6.15,0);
    \draw[->] (5,-1.2) -- (6.15,-0.6);
    \draw[->] (5,-1.2) -- (6.15,-1.8);

    \node[align=center, draw, fill=white, rounded corners=1pt, inner sep=2pt] at (6.8,0) {PoR-P2$_r$};
    \node[align=center, draw, fill=white, rounded corners=1pt, inner sep=2pt] at (6.8,-0.6) {PoR-P2$_r$};
    \node[align=center, draw, fill=white, rounded corners=1pt, inner sep=2pt] at (6.8,-1.2) {PoR-P2$_r$};
    \node[align=center, draw, fill=white, rounded corners=1pt, inner sep=2pt] at (6.8,-1.8) {PoR-P2$_r$};

    \node[align=center, draw, fill=white, rounded corners=1pt, inner sep=2pt] at (7.95,0) {P1$_{r+1}$};

    \draw[<->, blue, thick] (0.1,0.3) -- (1.8,0.3) node[midway, above, text height=1.5ex, text depth=0.25ex] {\small Fast Path};
    \draw[<->, blue, thick] (1.8,0.3) -- (7.45,0.3) node[midway, above, text height=1.5ex, text depth=0.25ex] {\small Base Path};
    \draw[<-, blue, thick] (7.45,0.3) -- (8.5,0.3) node[midway, above, text height=1.5ex, text depth=0.25ex] {\small Next Round};

    \end{tikzpicture}%
    }
    \caption{Fallback from Fast Path to Base Path}
    \label{fig:fallback_base_path}
    \vspace{-3mm}
    
\end{figure}

\subsection{\new{Speculative Execution over \PoA{}}}
\new{
A \PoA{} certificate in \sysname{} does not guarantee final commitment, but it does certify that a proposal has made recoverable partial progress. 
To avoid leaving this certified work idle during a partition, replicas may speculatively execute \PoA{}-backed proposals before they are committed via a \PoR{}. 
The execution result can be returned to the client as speculative output, though it does not constitute a final commit. 
After recovery, if the branch becomes the strongest branch, its speculative execution is reused on the committed path; 
otherwise, the speculative state is discarded, and its transactions must be re-proposed through the normal path. 
Speculative execution thus provides a meaningful way to measure partial progress.
%: the speculative throughput reported in our evaluation counts \PoA{}-backed proposals that have been locally ordered and executed, but not yet finally committed.
}

\section{PROOF SKETCH}
\label{sec:proof_sketch}

This section summarizes the main arguments behind \sysname's safety, liveness under partial synchrony, and partial liveness.
Full proofs are provided in Appendix~\ref{app:full_proofs}.

 \begin{restatable}[Safety]{theorem}{SafetyTheorem}\label{thm:safety}
No two correct replicas commit (finalize) conflicting proposals.
 \end{restatable}
\vspace{-2mm}
\begin{proof}[Proof sketch]
The argument follows from three ingredients.
First, any two strong quorums of size $2f{+}1$ intersect in at least $f{+}1$ replicas, and therefore share at least one correct replica.
Second, a correct replica casts at most one vote per round.
Hence, two conflicting proposals cannot both obtain valid \PoR{} certificates in the same round.

Across rounds, safety is preserved by the locking rule.
Once a correct replica observes a \PoR{} for a proposal, it locks that proposal (or a later descendant on the same branch) and will not vote for a conflicting one unless it observes a strictly higher-round \PoR{} on the competing proposal's certified ancestry.
Therefore, once a proposal becomes commit-ready under the Two-\PoR{} rule, the second \PoR{} contains at least $f{+}1$ correct voters that are locked on that branch.
By quorum intersection, any later strong quorum for a conflicting branch must intersect this locked set in at least one correct replica, which prevents the conflicting branch from gathering the certified chain needed for commitment.
\end{proof}

\subsection{Liveness}
 \begin{restatable}[Liveness (with probability 1)]{theorem}{LivenessTheorem}\label{thm:liveness}
During any sufficiently long stabilized synchronous period, correct replicas commit new proposals continuously;
over an unbounded sequence of such stabilized periods, commits occur infinitely often with probability~1.
 \end{restatable}
\begin{proof}[Proof sketch]
During a sufficiently long stabilized synchronous period, \new{the background timeout-calibration mechanism eventually makes every correct replica's local timeout baseline $\delta_i$ sufficiently large relative to the actual network delay bound.
}
Once local timeouts are sufficient, \PoR{}s and \RC{}s are delivered within the underlying synchronous delay bound.
Since every valid higher-round proposal carries one of these certificates as its \textit{enterCert}, certificate dissemination together with proposal-driven fast catch-up eventually re-aligns correct replicas into recurring \emph{good rounds}, in which all correct replicas enter the same round early enough to complete proposal exchange and vote collection.

In each good round, all correct replicas observe the same set of honest proposals.
If there is a unique strongest proposal, they all vote for it directly.
Otherwise, ties are broken by the common coin; with probability at least $\frac{2f{+}1}{n}$, the selected top-scoring proposer is honest, in which case all correct replicas converge on the same proposal and form a valid \PoR{}.
Since good rounds recur and each one succeeds with non-zero probability, \PoR{}s occur infinitely often.
Whenever two consecutive rounds produce extending \PoR{}s, the Two-\PoR{} rule commits the earlier proposal in the chain.
\end{proof}

\subsection{Partial Liveness}
 \begin{restatable}[Partial Liveness]{theorem}{PartialProgressTheorem}\label{thm:partial_progress}\sysname satisfies the \textbf{Partial Liveness} property.
Specifically, during network partitions, any connected component that contains at least $f{+}1$ correct replicas and whose internal message delays are bounded by $\Delta$\new{, and whose correct replicas use local timeout baselines satisfying $\delta_i \ge \Delta$, }continues to advance rounds locally and produce \PoA{}-certified partial progress.
After network recovery, this accumulated partial progress remains recoverable and can be incorporated into subsequent \PoR{}-backed commitment.
\end{restatable}
\vspace{-2mm}
\begin{proof}[Proof sketch]
Consider any connected component containing at least $f{+}1$ correct replicas and timely internal communication.
Within such a component, correct replicas can still exchange proposals within one local round and apply the same priority rule to identify a common strongest candidate.
As a result, the component can assemble a weak quorum and form a \PoA{}, thereby preserving certified partial progress even without a global strong quorum.

If no \PoR{} is formed in that round, correct replicas eventually time out and broadcast \textsc{WishNewRound}; since the component contains at least $f{+}1$ correct replicas, they can form an \RC{} for the next round and continue advancing locally.
Thus, under partition, \sysname preserves both round progression and recoverable partial progress.

A \PoA{} does not imply commitment, but it does imply data availability: at least one correct replica stores the certified proposal.
Once global connectivity is restored, a strong quorum becomes reachable again, and replicas resume forming \PoR{}s on the strongest recoverable branch.
Through the catch-up and back-fill mechanisms, previously accumulated \PoA{}-backed progress can be retrieved, extended, and incorporated into the committed prefix.
\end{proof}

\section{EVALUATION}
% To evaluate our partial progress conjecture in practice, we deploy our \sysname 
% protocol on Apache ResilientDB~\cite{resdb}. Our aim is to study the effects on throughout and latency upon network
% recovery to demonstrate partial progress as the system heals itself.

\begin{figure*}[t]
    \vspace{-3mm}
    \centering
    \includegraphics[width=\textwidth]{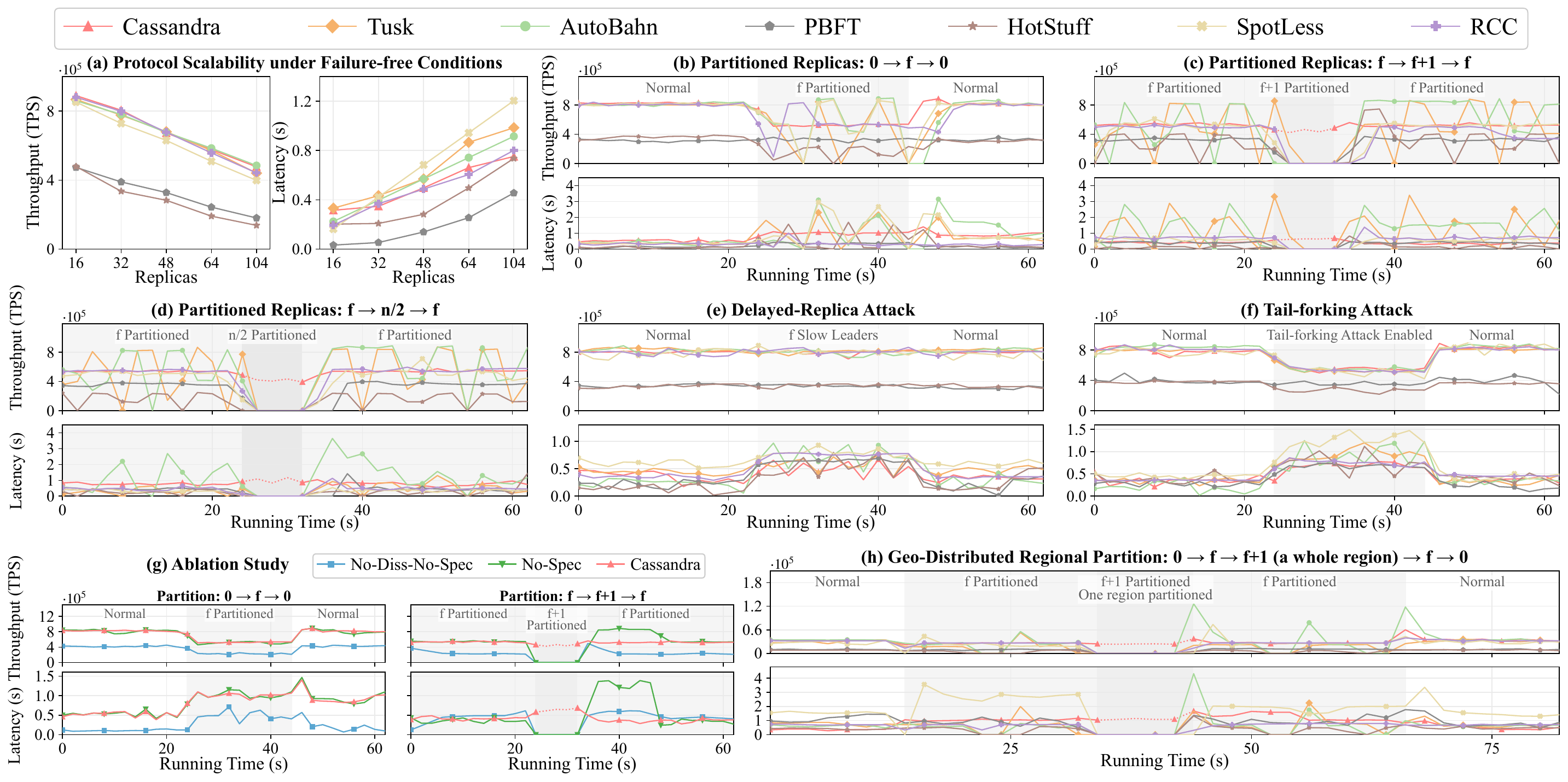}
    \vspace{-9mm}
    \newcaptionshepherd{Evaluation summary across scalability and partition scenarios}
    \vspace{-1mm}
    \label{fig:eval_protocols}
    \label{fig:bigpar-eval}
    \label{fig:weakpar-eval}
    \label{fig:weakpar-2-eval}
    \label{fig:weakpar-3-eval}
    \label{fig:weakpar-f-eval}
    \label{fig:ablation-study}
\end{figure*}

To validate our partial-progress claims, we implement \sysname{} and evaluate it against 
\new{several state-of-the-art BFT protocols: Tusk~\cite{narwhal-tusk}, AutoBahn~\cite{autobahn}, PBFT~\cite{pbft}, HotStuff~\cite{hotstuff}, SpotLess~\cite{spotless}, and RCC~\cite{rcc}.}
Our evaluation targets two key metrics---\textbf{throughput} and \textbf{latency}---and 
aims to address the following research questions:

\begin{itemize}[nosep,wide]
    \item \textbf{Q1:} \new{How does \sysname{} scale under stable network conditions compared with existing BFT protocols?} (\Cref{subsec:normal_result})

    \item \textbf{Q2:} \new{How does \sysname{} behave under network partitions, from $f$ partitioned replicas to balanced partitions?} (\Cref{subsec:partition_result})

    \item \textbf{Q3:} \new{How robust is \sysname{} under Byzantine behaviors such as intentional delay and tail-forking attacks?} (\Cref{subsubsec:byzantine-behavior})

    \item \textbf{Q4:} \new{What do the ablation study and latency breakdown reveal about \sysname{}'s performance?} (\Cref{subsubsec:ablation-study})

    \item \textbf{Q5:} \new{How does \sysname{} perform under geo-distributed regional partitions?} (\Cref{subsubsec:geo-eval})
\end{itemize}

\subsection{Implementation and Evaluation Setup}
\bheading{Implementation.}
We implement \sysname in C++~\cite{sourcecode} as a core consensus component within the open-source Apache ResilientDB platform~\cite{resdb,resdbpaper,resdbdemo}, a scalable blockchain infrastructure. 
% , utilizes \textit{std::thread} to manage asynchrony
% We use TCP to implement reliable point-to-point communication. 
\new{The prototype reuses ResilientDB's networking, batching, execution, monitoring, and key-management pipeline.
Messages are authenticated through ResilientDB's Crypto++ verifier, and our experiments use Ed25519 replica keys generated by ResilientDB's key-generation tool.
\PoA{}, \PoR{}, and \RC{} certificates are formed by collecting authenticated replica votes and checking corresponding thresholds.
Unless stated otherwise, \sysname{} enables the fast path, the dissemination layer, and speculative execution. 
Speculative execution lets replicas execute a \PoA{}-backed proposal before final \PoR{} commitment. 
This avoids long execution queues and reduces executor idleness during partitions, which speeds up recovery once the network is restored.
For experiments involving network partitions, we report \emph{speculative throughput}, 
which counts \PoA{}-backed proposals that have been locally ordered and executed but not yet finally committed.
}

\bheading{Evaluation Setup.}
We deploy \texttt{c5.4xlarge} instances (16 vCPUs, 32 GB memory) on AWS in Northern Virginia. 
For geo-distributed experiments, we use four regions: Northern Virginia, S\~ao Paulo, Frankfurt, and Singapore. 
Each replica is assigned a dedicated client so that all proposers propose at the same rate, and 
each client issues transactions in an open loop at a rate sufficient to saturate its replica. 
We also {\em batch transactions} for processing.
%
%In each setup, each node exclusively ran one instance of \sysname, with the number of clients matching the number of consensus processes to ensure that replicas were fully saturated.
For benchmarking, we use ResilientDB's key-value service, 
in which each transaction is a randomly generated key-value insertion of a fixed size (50 bytes). 
All reported numbers are averaged over three runs.

\begin{comment}
\begin{figure}[h]
    \centering
    \scalebox{0.82}{\ref{batchmain}}\\[3pt]
    \begin{tabular}{cc}
        \EvalBatchTPS &
        \EvalBatchLat
    \end{tabular}
    \vspace{-4mm}
    \caption{Evaluation of different batch sizes.} 
    \vspace{-2mm}
    \label{fig:eval_batch}
\end{figure}

\subsection{Impact of Batch Size}
In our initial experiments, 
we conducted tests using various batch sizes on 64 nodes 
to determine a suitable base size we will use in the following experiments, 
and the outcomes are depicted in Figure~\ref{fig:eval_batch}. 
The findings indicate that augmenting the batch size results in increased throughput. 
However, it also reveals a corresponding increase in latency with larger batch sizes.
The majority of protocols achieve peak throughputs, $5.2M$ tps, with a batch size of $800$. 
The latency for a batch size of $800$ is double that of a batch size of $400$, while the throughput remains similar for both sizes, $5.2M$ for the batch size of $800$ while $4.8M$ for the batch size of $400$. To achieve a balance between throughput and latency, we have decided to use a batch size of 400 as the default for our upcoming experiments.
\end{comment}

\subsection{Performance Under Normal Conditions} 
\label{subsec:normal_result}
First, we study the scalability of \sysname under stable network conditions (no failures).
\new{Figure~\ref{fig:eval_protocols}(a) varies the number of replicas from $16$ to $104$. 
All protocols exhibit a similar trend: as the number of replicas grows, throughput decreases and latency increases, 
because more participants must reach agreement. 
All protocols except PBFT and HotStuff sustain high throughput, as they are not constrained by a single-leader design. 
PBFT achieves higher throughput than HotStuff because its replicas can process transactions out of order, 
whereas HotStuff's rotating-leader design prohibits this. 
\sysname{} has a deterministic proposal-priority rule and dissemination, 
while Tusk, AutoBahn, RCC, and SpotLess employ similar multi-proposer or decoupled data paths.
\sysname{}'s partial-progress mechanisms, however, marginally increase latency relative to the lowest-latency baselines. 
At $104$ replicas, \sysname{}'s throughput is within $1.5\%$ of AutoBahn's.
}

\subsection{Performance under Network Partitions} \label{subsec:partition_result}
\new{
Next, we study the impact of network partitions on different protocols by increasing the number of partitioned replicas.
%Figure~\ref{fig:eval_protocols}(b)--(d) evaluates increasingly severe partition transitions, where the notation denotes the number of replicas partitioned away over time.
%The sequence is organized around the boundary between finality-capable and weak-quorum operation: the first transition keeps a $2f{+}1$ component, while the latter two remove such a component during the middle interval.
%We therefore distinguish \sysname{}'s \PoA{}-backed speculative ordering, AutoBahn's dissemination-layer progress, and strong-quorum ordering/finality, since disseminated or buffered data is useful after recovery but is not equivalent to recoverable ordering progress.
}

\subsubsection{\new{$0 \rightarrow f \rightarrow 0$ Partition}} 
\label{subsubsec:partition-f}

\new{
First, we partition only $f$ replicas (from $24$s to $44$s), leaving a connected component containing $2f{+}1$
replicas (Figure~\ref{fig:eval_protocols}(b)).
All protocols except PBFT lose throughput and gain latency during this partition;
PBFT remains relatively stable because its leader lies within the $2f{+}1$ connected component. 
HotStuff, Tusk, and AutoBahn degrade when their rotating leaders are partitioned away.
Tusk periodically completes some waves with a reachable leader, while AutoBahn's dissemination layer restores higher throughput after a brief stall in its HotStuff-style ordering layer.
Among the multi-leader protocols, RCC retains two-thirds of its steady-state performance after excluding the $f$ unreachable replicas, whereas SpotLess behaves similarly to AutoBahn because some of its instances select partitioned leaders.
\sysname{}, in contrast, continues forming \PoR{} certificates and finalizing transactions, although fewer connected replicas and switching between fast and base paths reduce its throughput by $33.7\%$.
}
%It sustains $545.6$ KTPS during the partition, retaining $66.3\%$ of its pre-partition throughput.
%The drop comes from fewer connected replicas and less proposal traffic, but \sysname{} does not suffer leader-placement stalls: every round can select the strongest available proposal by the deterministic priority rule, so a proposal from the partitioned-away side is simply not chosen.
%After the partition heals, all protocols return to their pre-partition operating range within about $4$\,s.}

\subsubsection{\new{$f \rightarrow f{+}1 \rightarrow f$ Partition}} 
\label{subsubsec:partition-fplusone}
\new{
Starting from an $f$-replica partition, we isolate one additional replica from $24$s to $32$s, leaving only
$2f$ replicas in the largest component (Figure~\ref{fig:eval_protocols}(c)).
All baselines lose ordering throughput because they require a strong quorum to run consensus.
\sysname{}, in contrast, observes only a $15\%$ drop in speculative throughput (and a $52\%$ increase in latency),
because \PoA{}s require just $f{+}1$ votes and thus continue to drive partial progress through speculative execution. 
Once the system returns to an $f$-replica partition, all protocols resume committing blocks; AutoBahn shows the largest throughput burst because its dissemination layer remained active while ordering was stalled.
%\PoA{}-backed proposals in \sysname can be folded into the \PoR{} path.
%instead sustains $446.3$ KTPS of \textit{speculative} throughput with $0.62$s latency because \PoA{} requires only $f{+}1$ votes and attaches weak-quorum evidence to an ordering branch.

%The remaining baselines require a $2f{+}1$ ordering quorum, so they stall or only buffer requests until a strong quorum returns.
%In this post-weak interval, the selected branch reaches $522.9$ KTPS with $0.36$s latency, showing that useful \PoA{}-backed work is carried into the finality path rather than regenerated from scratch.
%Speculative execution therefore absorbs most of the recovery burst: the selected branch has already been executed before final certification, while competing PoA branches are superseded.
}

\subsubsection{\new{$f \rightarrow \lfloor n/2 \rfloor \rightarrow f$ Partition}} 
\label{subsubsec:partition-half}
\new{
Finally, we partition roughly $50\%$ of the replicas from $24$s to $32$s, leaving a connected set of $\lceil n/2 \rceil$ replicas (Figure~\ref{fig:eval_protocols}(d)).
As in the previous experiment, all baselines lose ordering throughput because they cannot form a strong quorum.
In \sysname{}, the connected replicas can still generate \PoA{}s and speculatively execute transactions, 
though throughput falls by $21\%$ due to the smaller connected set. 
After restoration to an $f$-replica partition, all protocols behave as above, although \sysname{} incurs more reconciliation work than in \S\ref{subsubsec:partition-fplusone} because more replicas must catch up.
}

\subsection{\new{Performance under Byzantine Behavior}} 
\label{subsubsec:byzantine-behavior}
\new{
Next, we study the impact of Byzantine attacks on the different protocols. 
First, in Figure~\ref{fig:eval_protocols}(e), we consider a leader delay attack in which $f$ Byzantine replicas delay proposals and messages by 50\,ms. 
All protocols show only a marginal drop in throughput and return to steady state quickly, but they do suffer a large latency spike.
\sysname{} incurs only a $40\%$ latency increase during the delay interval, whereas HotStuff increases by nearly $300\%$.
This is because \sysname{} can switch between the fast and base paths when the leader is delayed.

%\sysname{} sustains about $800$ KTPS with $0.50$s latency.
%The delay introduces short-term noise to all protocols, but does not cause a sustained collapse.
%\sysname{} maintains the most stable throughput and latency because its proposal-priority rule allows it to select the strongest certified branch.
%Delayed proposals are thus likely superseded by stronger proposals from other replicas.
%Tusk, AutoBahn, RCC, and SpotLess maintain higher throughput because their data paths are more pipelined or decoupled, but they still show extra latency variance while waiting for delayed votes or batches.
%PBFT and HotStuff are more sensitive when delayed replicas lie on the leader path or the quorum-share collection path, so delay appears as lower throughput or sharper latency spikes rather than only transient noise.
}

\new{
Next, we evaluate the impact of a tail-forking attack (Figure~\ref{fig:eval_protocols}(f)). 
In a tail-forking attack, $f$ faulty replicas are strategically placed across rounds and leader positions to create short forks or unfinished tails. 
As a result, all protocols except PBFT experience a throughput drop, as proposals are dropped before they can commit, since the next leader may itself be faulty. 
PBFT's throughput is unaffected, since we do not designate its leader as faulty. 
\sysname{}'s latency is less volatile: 
whenever the leader is faulty, replicas fall back to the base path and select the strongest proposal. 
% This raises latency but keeps it predictable. 
In contrast, other protocols rely on their ordering layer to resolve these tails.

%\sysname{} sustains about $560$ KTPS with $0.66$s latency because its deterministic proposal-priority rule gives all correct replicas the same way to compare competing tails.
%Thus, even when Byzantine replicas create short or unfinished tails, correct replicas converge on the strongest certified branch and supersede weaker tails.
%PBFT and HotStuff are affected because a faulty tail can delay the next leader/quorum step or force replicas to abandon partially completed rounds.
%Tusk, AutoBahn, RCC, and SpotLess can continue moving data, but their ordering layer may still spend time resolving or bypassing weaker tails before the useful chain resumes.
%This explains why the high-throughput baselines remain active but show more latency volatility, whereas weaker tails in \sysname{} are deterministically superseded by the strongest certified branch.
}

\subsection{\new{Ablation Study and Latency Breakdown}} 
\label{subsubsec:ablation-study}
\new{
In Figure~\ref{fig:eval_protocols}(g), we isolate the effect of \sysname{}'s various optimizations under network partitions:
(i) {\em No-Diss-No-Spec} disables both the dissemination layer and speculative execution; and
(ii) {\em No-Spec} enables dissemination but disables speculative execution.
%These are incremental implementation variants; leaderless selection, \PoA{}/\PoR{} certification, and the certificate-driven pacemaker remain enabled because disabling them would change the protocol rather than isolate an optimization.

In the $0 \rightarrow f \rightarrow 0$ case, 
a strong quorum remains reachable, so the real impact comes from enabling or disabling the dissemination layer.
Enabling dissemination improves partition-period throughput by $2.27{\times}$, 
because unselected proposals remain available and can be reused rather than discarded as wasted work.

In the $f \rightarrow f{+}1 \rightarrow f$ case, 
the middle interval lacks a strong quorum, so this experiment isolates the benefit of speculative execution. 
Without it, the {\em No-Diss-No-Spec} and {\em No-Spec} variants yield zero speculative throughput during the interval in which $f{+}1$ replicas are partitioned ($24$s to $32$s). 
The dissemination layer's benefit to ordered throughput appears once a strong quorum is reachable, since transactions are already available for ordering.
}

\new{\bheading{Latency Breakdown.}
Next, we dissect \sysname{}'s latency to understand where it actually spends time. Roughly $~70\%$ of \sysname{}'s 
latency goes to forming a \PoR{} certificate---which requires collecting votes and disseminating the certificate---and to committing, which entails forming two consecutive \PoR{}s. 
Proposal dissemination and execution each contribute $10\%$, 
strongest-proposal election $6\%$, cryptography only $1\%$, and other overheads $3\%$. 
% Thus, cryptography and proposal selection are not the bottleneck; 
% latency is dominated by the commit path and by communication and certificate coordination.
}

\subsection{\new{Geo-Distributed Regional Partition}} 
\label{subsubsec:geo-eval}

\new{
We study the impact in a geo-distributed setting in Figure~\ref{fig:eval_protocols}(h): 
we deploy $n=31$ replicas across four regions, with regional group sizes of $8/8/8/7$. 
The experiment follows the path: normal connectivity $\rightarrow$ $f$ partitioned replicas $\rightarrow$ $f{+}1$ partitioned replicas $\rightarrow$ $f$ partitioned replicas $\rightarrow$ recovery. 
We distribute the $f$ partitioned replicas across the four regions as $7/1/1/1$, and 
the $f{+}1$ partitioned replicas as $8/1/1/1$, which renders one full region unavailable. 

WAN latency, bandwidth heterogeneity, and cross-region quorum communication reduce throughput and increase latency for all protocols, but the qualitative behavior matches the partition results above. 
During the $f$-partitioned phase, a quorum of $2f{+}1$ replicas remains available, 
so all protocols continue committing proposals, albeit with a decrease similar to that in \S\ref{subsubsec:partition-f}.
During the $f{+}1$-partitioned phase, all protocols except \sysname{} are unable to commit any blocks, 
whereas \sysname{} continues to speculatively execute blocks. 
After reconnection, cross-region quorum paths are restored and all protocols return to their steady-state performance.

}

\section{RELATED WORK}
\label{sec:related_work}
The design of \sysname draws inspiration from and builds upon several decades of research in BFT consensus. 
We discuss network partitions, view synchronization, dual-path protocols, sharding, concurrent and multi-leader BFT, and asynchronous protocols.

\bheading{Network Partition.}
The challenge of maintaining consistency and availability under network partitions is fundamentally tied to the CAP theorem.
Traditional BFT protocols typically prioritize safety and stall indefinitely when a strong quorum ($2f{+}1$) cannot be assembled.
\textit{Partitionable Blockchain Consensus}~\cite{hood2020partitionable} was among the first to formalize this problem.  
It proves impossibility in a fully asynchronous model, then leverages external \textit{partition detectors} to coordinate the partition.  
Raptr~\cite{tonkikh2025raptr} introduces prefix consensus: replicas maintain safety by finalizing only the common transaction prefix. 
Then it reconciles divergent suffixes upon restoration.
More recently, AutoBahn~\cite{autobahn} preserves a limited form of partial progress under network disruptions, but this progress is largely confined to the dissemination layer rather than the ordering layer, and thus does not directly address \textit{designated-leader dependence}.
% In particular, data lanes may continue advancing, while ordering still depends on a leader-driven consensus step, which leaves it vulnerable to leader-mediated inclusion bias (e.g., selective omission). 

In contrast to these approaches, \sysname{} neither relies on external partition detectors nor confines partial progress to data dissemination alone.
Instead, its core ordering path allows partitions to independently order transactions that are later integrated through an implicit merge method.

\bheading{View Synchronization.}
View synchronization is fundamental for aligning replicas within a common round to ensure sustained consensus progress.
While classic protocols like PBFT~\cite{pbft} utilize an exponential back-off strategy for eventual convergence, HotStuff~\cite{hotstuff} decouples this logic into a dedicated \textit{pacemaker}.
Building on this modularity, state-of-the-art mechanisms have evolved to address the quadratic communication overhead inherent in traditional designs. Notable advancements include Lewis-Pye’s hierarchical, epoch-based framework~\cite{lewispye2022quadratic,fever}, which achieves amortized $O(n)$ communication complexity, and Lumiere~\cite{lewispye2024lumiere}, which further optimizes steady-state performance with optimistic responsiveness.

While these advancements significantly enhance steady-state performance, they remain fundamentally dependent on a $2f{+}1$ strong quorum of ``new-view'' messages to transition between rounds, leading to total system stalls during network partitions. 
Cassandra fills this operational void by introducing a novel pacemaker that facilitates round advancement through the evidence from $f{+}1$ replicas, thereby sustaining partial progress where other state-of-the-art pacemakers fail to coordinate.

\new{\bheading{Dual-Path Design BFT.}
A common theme in practical BFT systems is to expose an \textit{optimistic fast path} that achieves low latency and linear communication when the network is stable and the active leader is responsive, while retaining a \textit{fallback path} that preserves safety under faults and delays. 
}

\new{Dual-path protocols such as Zyzzyva, FaB, and SBFT~\cite{zyzzyva,fab,sbft} use an optimistic fast path for low latency and a fallback path for faults.
These designs are effective in stable networks, but the common path and recovery path remain largely leader-driven.
DAG-based Bullshark~\cite{spiegelman2022bullshark} also uses a leader in the synchronous path and a coin-assisted fallback for convergence.
Recent work targets the tail latency of leader-based protocols, where a crashed or slow leader wastes an entire view.
Hermes~\cite{hermes} is a rotating-leader protocol for $n=5f+1$ that finalizes the leader's proposal in two rounds in the good case, while votes for a fallback proposal at the timeout still finalize the heaviest common prefix, so expired views are not wasted.
Ambulance~\cite{ambulance} replaces timeouts with protocol-rigged races, in which replicas race each other by executing protocol steps rather than waiting on a clock, retaining leader-based common-case performance while recovering from slowdowns as quickly as cooperative asynchronous protocols.
Both nonetheless organize each view around a designated leader whose proposal is the fast-path outcome.
\sysname{} uses randomness only as a tie-breaker among otherwise indistinguishable proposals; ordering is driven by the proposal priority rather than designated-leader or coin-driven leader election.}

\new{As a result, \sysname{} removes designated-leader (and coin-driven leader-election) dependence from its core ordering path via partitionable leader selection and deterministic proposal priority, enabling components to continue making order-relevant partial progress even when global leader reachability or coin-driven coordination becomes ineffective under partitions.}

\bheading{Sharding Protocols.}
Sharded systems aim to improve scalability by partitioning replicas into smaller committees that process disjoint subsets of state or transactions in parallel.  
Early systems such as Chainspace~\cite{chainspace2018} explored object-centric sharding with BFT committees and atomic object locking.
Building on this, OmniLedger~\cite{omniledger2018} utilized secure randomness for epoch-based committee assignment, while RapidChain~\cite{rapidchain} pipelines block production across shards.
% assigns clients and validators to epochs via secure randomness and runs a ByzCoinX-style BFT within each shard; atomic cross-shard payments are finalised with an optimistic commit–abort protocol.
% and uses erasure-coded block dissemination to reduce cross-committee bandwidth, achieving near-linear throughput under partial synchrony.  
Cerberus~\cite{cerberus2021} further evolved this by employing per-shard DAGs to eliminate single-leader bottlenecks.
Extending this line of work, ByShard~\cite{byshard} optimizes this paradigm by introducing a sharded database framework that efficiently integrates BFT consensus with atomic cross-shard coordination. RCanopus~\cite{keshav2019rcanopus} geographically divides replicas into Byzantine Groups (BGs) that run local BFT, while a super-consensus orders the BG outputs.

While these systems significantly improve horizontal scalability, they still inherently stall if fewer than $2f{+}1$ replicas remain in a shard after a partition. 
\sysname’s partial-progress technique is \textit{orthogonal}. 
It could serve as an intra-shard consensus component, allowing each committee to process local transactions with non-zero throughput even when a partition leaves only $f{+}1$ honest replicas connected.

\bheading{Concurrent and Multi-Leader BFT.}
Recent work~\cite{rcc,spotless,iss,ladon,orthrus,hydra,mirbft,dqbft} has also explored allowing multiple replicas to concurrently drive consensus, which helps reduce the bottleneck of a single designated leader under stable conditions.
At a high level, this design may appear to mitigate \textit{leader dependence}, since progress is no longer tied to one pre-assigned proposer.
However, these protocols still require global coordination to reconcile concurrent proposals and remain constrained by the \textit{strong quorum dependence}.
While they improve steady-state throughput, they do not directly address the partition setting considered in \sysname: failing to make useful progress during network partitions.

\bheading{Asynchronous Protocols.}
Fully asynchronous protocols, such as HoneyBadgerBFT~\cite{miller2016honey} and BEAT~\cite{duan2018beat}, eliminate timing assumptions by combining reliable broadcast with threshold cryptography.
Recent advancements in DAG-based BFT, including DAG-Rider~\cite{keidar2021dagrider} and Tusk~\cite{narwhal-tusk}, further optimize performance by decoupling data dissemination from ordering, achieving optimal resilience and amortized $O(n^2)$ complexity.
Fides~\cite{xie2025fides} further enhances this model by offloading critical operations into trusted hardware for an optimized threshold and better performance.

Despite their resilience to asynchrony, these protocols remain bound by the \textit{Strong Quorum Dependence}: none can make progress once fewer than a quorum of replicas are reachable, since they still require a classical quorum to form certificates.
In contrast, \sysname addresses this by allowing a weaker quorum to construct weaker certificates, achieving non-zero progress during weak quorums---a breakthrough from 0 to 1.

\section{CONCLUSION}
We presented \sysname{}, a BFT protocol that \new{breaks the all-or-nothing progress model of traditional consensus by enabling partial, safety-preserving progress} even when no strong quorum is reachable.
\new{\sysname{} combines two-tier certification that separates recoverable progress from finality, a deterministic proposal-priority rule for partition-local ordering and implicit reconciliation, and a decoupled pacemaker for certificate-driven round advancement.
Together, these mechanisms enable isolated partitions to accumulate useful ordering progress that is implicitly reconciled once the network heals. 
Our evaluation shows that \sysname{} remains competitive with state-of-the-art BFT protocols under stable conditions, while delivering continuous partial throughput during partitions.
}

\begin{acks}
This work is partially funded by NSF Award Number 2245373.
\end{acks}

\clearpage
\bibliographystyle{ACM-Reference-Format}
\renewcommand{\refname}{REFERENCES}
% Start balancing in the first column of the final bibliography page.
\makeatletter
\pretocmd{\@lbibitem}{\ifstrequal{#2}{ladon}{\balance}{}}{}%
  {\PackageError{cassandra}{Could not set bibliography balancing point}{}}
\makeatother
\bibliography{reference}
\clearpage
\nobalance

\appendix
\setcounter{theorem}{0}
\section{CORRECTNESS PROOF}
\label{app:full_proofs}

This section provides the full proofs of Theorems~\ref{thm:safety}, \ref{thm:liveness}, and~\ref{thm:partial_progress}.

\subsection{Safety}
\label{sec:safety}

We assume a Byzantine fault model with $n=3f{+}1$ replicas, where at most $f$ replicas are Byzantine.
A \emph{strong quorum} requires $Q=2f{+}1$ votes, and a \emph{weak quorum} requires $Q_{\mathrm{weak}}=f{+}1$ votes.
We say two proposals \emph{conflict} if neither is an ancestor of the other.

\begin{lemma}[Quorum Intersection]\label{lem:quorum_intersection}
Any two strong quorums of size $2f{+}1$ intersect in at least $f{+}1$ replicas; in particular, they share at least one correct replica.
\end{lemma}

\begin{proof}
Let $Q_1$ and $Q_2$ be two strong quorums with $|Q_1|=|Q_2|=2f{+}1$.
Since $n=3f{+}1$,
\[
|Q_1\cap Q_2| \ge |Q_1|+|Q_2|-n = (2f{+}1)+(2f{+}1)-(3f{+}1)=f{+}1.
\]
As at most $f$ replicas are Byzantine, $Q_1\cap Q_2$ contains at least one correct replica.
\end{proof}

\begin{lemma}[No Double Vote]\label{lem:no_double_vote}
A correct replica votes for at most one proposal in a given round.
\end{lemma}

\begin{proof}
Correct replicas follow the protocol.
In each round $r$, a correct replica selects a single candidate proposal according to the priority rule and broadcasts \textsc{Vote} messages for exactly that proposal.
\end{proof}

\begin{lemma}[\PoR{} Uniqueness]\label{lem:unique_cert}
For any round $r$, there cannot exist two valid \PoR{} certificates that certify conflicting proposals.
\end{lemma}

\begin{proof}
Assume for contradiction that there exist two valid \PoR{}s in round $r$ certifying conflicting proposals $b$ and $b'$.
Each \PoR{} requires $2f{+}1$ votes.
By Lemma~\ref{lem:quorum_intersection}, the two voter sets intersect in at least one correct replica.
That replica would have voted for both $b$ and $b'$ in round $r$, contradicting Lemma~\ref{lem:no_double_vote}.
\end{proof}

\begin{lemma}[Lock Monotonicity]\label{lem:lock_monotonicity}
For any correct replica $p_i$, its lock $lock_i$ is monotonically non-decreasing in round number.
Moreover, once $p_i$ observes a \PoR{} for a proposal, it locks that proposal and will not vote for any conflicting proposal unless it later observes a strictly higher-round \PoR{} on the competing proposal's certified ancestry.
\end{lemma}

\begin{proof}
Whenever $p_i$ observes a valid \PoR{}$_r$ for proposal $b_r$, it updates $lock_i$ to $b_r$ unless it is already locked on a descendant of $b_r$.
Thus $\mathsf{round}(lock_i)$ can only increase over time.

After locking, \textsc{SafeToVote} allows $p_i$ to vote only for a proposal that either
(i) extends the current locked branch, or
(ii) carries a strictly higher-round \PoR{} that justifies switching.
By protocol validity, any such higher-round \PoR{} must lie on the proposal's certified ancestry.
Therefore, $p_i$ never votes for a conflicting proposal unless it has first observed higher certified evidence that justifies the switch.
\end{proof}

\begin{lemma}[Commit Implies a Locked Set]\label{lem:commit_implies_locked_set}
If a proposal $w=b_r$ is committed by the Two-\PoR{} commit rule via a certified chain
$b_r \leftarrow b_{r+1}$ with \PoR{}s $\PoR_r$ and $\PoR_{r+1}$,
then there exists a set $L$ of at least $f{+}1$ correct replicas such that every replica in $L$ is locked on $w$ (or its descendants).
\end{lemma}

\begin{proof}
Let $Q_{r+1}$ be the voter set that formed $\PoR_{r+1}$ for $b_{r+1}$.
Since $|Q_{r+1}|=2f{+}1$ and at most $f$ replicas are Byzantine, $Q_{r+1}$ contains at least $f{+}1$ correct replicas; let $L$ be those correct voters.

By protocol validity, a correct replica votes for $b_{r+1}$ only after verifying that it extends $b_r$ and that its certified justification includes the parent certificate $\PoR_r$ on $b_r$ (carried either directly as the round-entry certificate or as the highest justified \PoR{} on $b_{r+1}$'s certified ancestry).
Hence, every replica in $L$ has observed $\PoR_r$ before voting for $b_{r+1}$.
By Lemma~\ref{lem:lock_monotonicity}, each such replica locks on $w=b_r$ (or a descendant on the same branch).
\end{proof}

\SafetyTheorem*
\begin{proof}[Proof of Theorem~\ref{thm:safety}]
Assume for contradiction that two correct replicas commit conflicting proposals $w$ and $b$, with $\mathsf{round}(w) < \mathsf{round}(b)$.

Since $w$ is committed, by Lemma~\ref{lem:commit_implies_locked_set}, there exists a set $L$ of at least $f{+}1$ correct replicas locked on $w$ (or its descendants).

Since $b$ is committed, there must exist at least one \PoR{} that certifies a proposal conflicting with $w$ in some round after $\mathsf{round}(w)$.
Let $\PoR^\star$ be the earliest such \PoR{}, certifying a proposal $x^\star$ that conflicts with $w$, and let $Q^\star$ be its voter set.

Because $|Q^\star|=2f{+}1$ and $|L|=f{+}1$, we have
\[
|Q^\star \cap L| \ge (2f{+}1)+(f{+}1)-(3f{+}1)=1.
\]
Thus, $Q^\star$ intersects $L$ in at least one correct replica $p$.

Replica $p$ is locked on the branch of $w$.
For $p$ to vote for $x^\star$, \textsc{SafeToVote} requires that $x^\star$ either extends $p$'s locked branch or carries a strictly higher-round \PoR{} on $x^\star$'s certified ancestry that justifies switching.
However, $x^\star$ conflicts with $w$, so it does not extend the locked branch.
Moreover, any such higher-round \PoR{} on $x^\star$'s ancestry would itself certify a conflicting proposal earlier than $\PoR^\star$, contradicting the choice of $\PoR^\star$ as the earliest conflicting \PoR{} after $w$.
Therefore, $p$ cannot vote for $x^\star$, contradicting $p \in Q^\star$.

Hence, no two correct replicas can commit conflicting proposals.
\end{proof}

\subsection{Liveness} \label{sec:liveness}

\begin{definition}[Timeout-Sufficient Replica]\label{def:timeout_sufficient}
\new{Fix a stabilized synchronous interval whose actual message delay bound is $\Delta$.
A correct replica $p_i$ is \textit{timeout-sufficient} during that interval if its local timeout baseline satisfies $\delta_i \ge \Delta$.
}\end{definition}

\begin{lemma}[Eventual Timeout Sufficiency]\label{lem:eventual_timeout_sufficiency}
\new{During any sufficiently long stabilized synchronous interval in which at
least $2f{+}1$ correct replicas are mutually connected, every correct replica
eventually becomes timeout-sufficient.
}\end{lemma}

\begin{proof}
\new{Fix such a stabilized synchronous interval with actual delay bound $\Delta$.
Consider any correct replica $p_i$.
}

\new{If $p_i$ cannot yet collect enough $\langle \mathsf{SYNC\mbox{-}READY}, \nu \rangle$ messages to locally form a synchronization certificate, it keeps retransmitting its $\langle \mathsf{SYNC\mbox{-}READY}, \nu \rangle$ message.
Since at least $2f{+}1$ correct replicas are mutually connected and all messages among them are delivered within $\Delta$, eventually every correct replica receives $2f{+}1$ valid \textsc{SyncReady} messages for some synchronization view and forms a local $\langle \mathsf{SYNC\mbox{-}CERT}, \nu \rangle$.
}

\new{If replica $p_i$ locally forms a synchronization certificate but does not receive certificates from at least $f{+}1$ distinct replicas within time $\delta_i$, it doubles its $\delta_i$ and retries in a higher synchronization view.
Hence, as long as $\delta_i < \Delta$, repeated failures cause $\delta_i$ to grow exponentially.
}

\new{Once $\delta_i \ge \Delta$, after $p_i$ forms its local synchronization certificate, every correct replica that is already in the same synchronization view can form and broadcast its own synchronization certificate within at most $\Delta$, and that certificate is delivered back to $p_i$ within at most another $\Delta$.
Therefore, replica $p_i$ receives remote synchronization certificates from at least $f{+}1$ distinct replicas within time $\delta_i$, so the calibration attempt succeeds.
}

\new{Thus, after finitely many retries, every correct replica reaches a timeout baseline $\delta_i \ge \Delta$ and remains timeout-sufficient throughout the stabilized interval.
}\end{proof}

\begin{definition}[Good round]\label{def:good_round}
A round $r$ is \emph{good} if there exists a stabilized synchronous interval of length $\ge3\Delta$, partitioned into three consecutive subintervals $I_1$, $I_2$, and $I_3$ each of length $\ge\Delta$, such that:
(i) by the end of $I_1$, all correct replicas have entered round $r$ and are timeout-sufficient;
(ii) during $I_1 \cup I_2$, every proposal broadcast by a correct replica in round $r$ is delivered to every correct replica; and
(iii) during $I_3$, all correct replicas remain in round $r$, and every vote broadcast by a correct replica in round $r$ is delivered to every correct replica.
\end{definition}

\begin{lemma}[Good round exists]\label{lem:good_round_exists}
During any sufficiently long stabilized synchronous period, the pacemaker ensures that all correct replicas eventually synchronize their round counters and enter good rounds.
\end{lemma}

\begin{proof}
Consider a stabilized synchronous period in which messages between correct replicas are delivered within $\Delta$.
\new{By Lemma~\ref{lem:eventual_timeout_sufficiency}, after finitely many retries, every correct replica becomes timeout-sufficient.
}If any correct replica enters a higher round using a valid \PoR{} or \RC{}, it broadcasts that certificate.
If a lagging correct replica has not yet seen the standalone certificate, it will still receive a valid higher-round proposal carrying the same certificate as its \textit{enterCert}.
In either case, within at most one message delay, every correct replica receives round-entry evidence for that higher round and can catch up.

Moreover, when a lagging replica catches up to an already active round, it joins that round using only the remaining local waiting needed before entering the election phase, rather than restarting a full fresh round.
Thus, during a sufficiently long stabilized synchronous period, the skew among correct replicas is continuously reduced, and there exists a round $r$ in which all correct replicas enter early enough to complete the proposal-exchange window ($>2\Delta$) and the election window ($>\Delta$) together.
Hence, round $r$ is good by Definition~\ref{def:good_round}.

As the stabilized synchronous period continues, the same argument applies repeatedly, so good rounds recur.
\end{proof}

\begin{lemma}[\PoR{} in a good round (probabilistic)]\label{lem:por_in_good_round}
In any good round $r$, honest replicas form a valid \PoR{} for a single proposal with probability at least $\frac{2f{+}1}{n}$.
\end{lemma}

\begin{proof}
Fix a good round $r$.
By Definition~\ref{def:good_round}, by the end of the proposal-exchange window ($I_1 \cup I_2$), every correct replica has received every proposal broadcast by a correct replica in round $r$.
We assume that each replica accepts at most one proposal per author per round, ignoring additional proposals from the same author as equivocations.
Hence, the number of candidate proposals considered in round $r$ is at most $n$.

Let $\mathcal{M}_r$ denote the set of proposals with maximum log strength under $\mathcal{S}(\cdot)$ among the proposals received by a correct replica in round $r$.
If $|\mathcal{M}_r|=1$, all correct replicas select the same strongest proposal and vote for it.

Otherwise, ties are broken using the common coin $\mathcal{C}_r$ for round $r$.
Since $\mathcal{C}_r$ is obtained by threshold aggregation of $f{+}1$ signature shares from $f{+}1$ distinct replicas, all correct replicas compute the same $\mathcal{C}_r$.
For each proposal $\mathcal{P}\in \mathcal{M}_r$, replicas compute
\[
\mathrm{Score}(\mathcal{P}) = H(\mathcal{C}_r \,\|\, \mathcal{P}.\textsf{proposer\_id})
\]
and select the proposal with the highest score.

Consider the event that the highest score among the tied proposers corresponds to an honest proposer.
Since proposer identities are fixed and $\mathcal{C}_r$ is unpredictable to the adversary before round $r$, the top-scoring proposer is uniformly distributed over the $n$ proposers, and is honest with probability $\frac{2f{+}1}{n}$.
Conditioned on this event, all correct replicas select the same honest proposal.

During the election window $I_3$, all correct replicas broadcast their \textsc{Vote}s, and by Definition~\ref{def:good_round}, those votes are delivered to every correct replica within one message delay.
Therefore, the selected proposal receives at least $2f{+}1$ votes, allowing correct replicas to assemble a valid \PoR{} in round $r$.
\end{proof}

\LivenessTheorem*

\begin{proof}[Proof of Theorem~\ref{thm:liveness}]
Consider a sufficiently long stabilized synchronous period.
By Lemma~\ref{lem:eventual_timeout_sufficiency}, sufficiently long stabilized synchronous periods first make all correct replicas timeout-sufficient; then, by Lemma~\ref{lem:good_round_exists}, good rounds recur.
In each good round, Lemma~\ref{lem:por_in_good_round} implies that a \PoR{} is formed with non-zero probability.

Assuming fresh common-coin outputs across rounds, these per-round success events occur with a fixed non-zero probability, so over an unbounded sequence of good rounds, \PoR{}s occur infinitely often with probability~1.

Whenever two consecutive rounds produce \PoR{}s for proposals that extend one another, the Two-\PoR{} commit rule commits the earlier proposal in the two-chain.
Hence, during any sufficiently long stabilized synchronous period, correct replicas commit new proposals continuously; over an unbounded sequence of such stabilized periods, commits occur infinitely often with probability~1.
\end{proof}

\subsection{Partial Liveness}

\sysname distinguishes between \emph{commitment} (which requires strong quorums and \PoR{}s) and \emph{partial progress} (which preserves useful work and availability evidence during partitions).

\begin{lemma}[Weak Quorum Progress]\label{lem:poa_in_weak_component}
In any connected component that contains at least
$Q_{\mathrm{weak}}=f{+}1$ honest replicas, whose internal message delays are bounded by $\Delta$, \new{and in which all participating correct replicas are timeout-sufficient (i.e., they have local timeout baselines satisfying $\delta_i \ge \Delta$), }the component can form at least one \PoA{} in a round.
\end{lemma}

\begin{proof}
Let $\mathcal{C}$ be a connected component containing $k \ge f{+}1$ honest and timeout-sufficient replicas.
Within $\mathcal{C}$, all honest replicas receive each other's proposals within the proposal-exchange window ($>2\Delta$).
By applying the deterministic \textit{Priority Rule}, these replicas converge on the same strongest candidate and broadcast votes for it during the subsequent election window ($>\Delta$).
Since $|\mathcal{C}| \ge Q_{\mathrm{weak}} = f{+}1$, replicas are guaranteed to aggregate the required weak quorum of votes.
Consequently, replicas in $\mathcal{C}$ can locally assemble a valid \PoA{}.
\end{proof}

\begin{lemma}[RC Progress in a Weak Component]\label{lem:rc_in_weak_component}
In any connected component that contains at least
$Q_{\mathrm{weak}}=f{+}1$ honest replicas , the component can form a valid \RC{} for round $r+1$.
\end{lemma}

\begin{proof}
If no \PoR{} is formed in round $r$, then every correct replica in the component eventually reaches the election timeout for round $r$ and broadcasts a \textsc{WishNewRound} $(r)$ message.
Since the component contains at least $f{+}1$ honest replicas and those messages are eventually delivered (without strict timing guarantees), replicas in the component can collect at least $f{+}1$ such messages.
By the protocol rule for round synchronization, these messages form a valid \RC{} for round $r+1$.
\end{proof}

\begin{lemma}[Data Availability]\label{lem:poa_implies_availability}
If a proposal is certified by a \PoA{} or a \PoR{}, its data is available at at least one correct replica.
Hence, any correct replica can lazily retrieve it via the catch-up mechanism when connectivity permits.
\end{lemma}

\begin{proof}
A \PoA{} or a \PoR{} requires at least $f{+}1$ votes, which guarantees that at least one correct replica has verified and stored the proposal.
Hence, this certificate ensures that the proposal remains retrievable in the system.
\end{proof}

\PartialProgressTheorem*

\begin{proof}[Proof of Theorem~\ref{thm:partial_progress}]
Let $\mathcal{C}$ be a connected component containing at least $f{+}1$ correct and timeout-sufficient replicas and timely internal communication.

By Lemma~\ref{lem:poa_in_weak_component}, in each local round, replicas in $\mathcal{C}$ can form a \PoA{} for the strongest locally visible proposal, thereby preserving certified partial progress.
If the component cannot form a strong certificate in that round, then by Lemma~\ref{lem:rc_in_weak_component}, the same replicas can still form an \RC{} for the next round and continue advancing locally.
Thus, even without a global strong quorum, $\mathcal{C}$ continues both to make round progress and to accumulate recoverable \PoA{}-certified work.

A \PoA{} does not imply commitment, but by Lemma~\ref{lem:poa_implies_availability}, it guarantees that the certified proposal remains available at some correct replica.
Once global connectivity is restored and a strong quorum $Q=2f{+}1$ becomes reachable again, replicas exchange their local histories, compare the available branches via the deterministic priority rule, and resume forming \PoR{}s on the strongest recoverable branch.
The catch-up and back-fill mechanisms then retrieve and extend previously accumulated \PoA{}-backed proposals on that branch, allowing the system to reincorporate partition-time partial progress into subsequent \PoR{}-backed commitment.

Therefore, \sysname satisfies Partial Liveness.
\end{proof}

\end{document}